\documentclass[11pt]{article}
\usepackage{graphicx}
\usepackage{psfrag,cite,amsmath,graphicx,epsfig,latexsym,subfigure,color}
\usepackage[scriptsize,bf]{caption}
\usepackage{color,xcolor,multirow}
\usepackage{amssymb, amsthm,bm}
\usepackage[boxruled]{algorithm}
\usepackage{algpseudocode}
\usepackage{setspace}
\usepackage{array}
\usepackage{booktabs}
\usepackage{hyperref}
\usepackage[normalem]{ulem}
\usepackage{mathtools}
\usepackage{etoolbox}

\newtheorem{theorem}{Theorem}
\newtheorem{claim}{Claim}
\newtheorem{lemma}{Lemma}
\newtheorem{remark}{Remark}
\newtheorem{assumption}{Assumption}
\usepackage{xparse}
\NewDocumentCommand{\ShowInline}{v}{%
#1%
}
\newcommand{\bx}{\bold{h}}
\newcommand{\ba}{\bold{a}}
\newcommand{\by}{\bold{y}}
\newcommand{\bz}{z}

\newcommand{\bh}{\bold{h}}
\newcommand{\bp}{\bold{p}}

\newcommand{\dg}{\bar{g}^t}
\newcommand{\df}{\bar{f}^t}
\newcommand{\dL}{\bar{L}}
\newcommand{\dC}{\bar{C}}

\newcommand{\bF}{F^t}

\def\bblambda{\boldsymbol  \lambda}

\newcommand{\bTheta}{\mathbf{\Theta}}
\usepackage[algo2e,linesnumbered,vlined,ruled]{algorithm2e}
\DeclarePairedDelimiter{\dotp}{\langle}{\rangle}

\SetCommentSty{mycommfont}

\begin{document}
\title{Learning to Continuously Optimize Wireless Resource in a Dynamic Environment: A Bilevel Optimization Perspective}

\author{Haoran Sun, Wenqiang Pu, Xiao Fu, Tsung-Hui Chang,  and Mingyi Hong \thanks{H. Sun and M. Hong are with Department of Electrical and Computer Engineering, University of Minnesota, Minneapolis, MN 55455, USA. W. Pu is with the Shenzhen Research Institute of Big Data, Shenzhen, China.  T-H. Chang is with Shenzhen Research Institute of Big Data, The Chinese University of Hong Kong, Shenzhen, China.   X. Fu is with the School of Electrical Engineering and Computer Science, Oregon State University, Corvallis, OR 97331, USA. H. Sun and M. Hong are supported by NSF/Intel MLWiNS: CNS-2003033. X. Fu is supported by NSF/Intel MLWiNS: CNS-2003082. T-H. Chang is supported by the NSFC, China, under Grant 62071409 and 61731018. \newline A short version of this paper \cite{sun2021} has been published in the IEEE International Conference on Acoustics, Speech, \& Signal Processing (ICASSP), 2021.}
}

\maketitle
\begin{abstract}
There has been a growing interest in developing data-driven, and in particular deep neural network (DNN) based methods for modern communication tasks. For a few popular tasks such as power control, beamforming, and MIMO detection, these methods achieve state-of-the-art performance while requiring less computational efforts, less resources for acquiring channel state information (CSI), etc.  However, it is often challenging for these approaches to learn in a dynamic environment.

This work develops a new approach that enables data-driven methods to continuously learn and optimize resource allocation strategies in a dynamic environment. Specifically, we consider an ``episodically dynamic" setting where the environment statistics change in ``episodes", and in each episode the environment is stationary.  We propose to build the notion of continual learning (CL) into wireless system design, so that the learning model can incrementally adapt to the new episodes, {\it without forgetting} knowledge learned from the previous episodes. Our design is based on a novel bilevel optimization formulation which ensures certain ``fairness"  across different data samples. We demonstrate the effectiveness of the CL approach by integrating it with two popular DNN based models for power control and beamforming, respectively, and testing using both synthetic and ray-tracing based data sets.  These numerical results show that the proposed CL approach is not only able to adapt to the new scenarios quickly and seamlessly, but importantly, it also maintains high performance over the previously encountered scenarios as well.  
\end{abstract}


\section{Introduction} \label{introduction}
Deep learning (DL) has been successful in many applications such as computer vision \cite{voulodimos2018deep}, 
natural language processing \cite{young2018recent}, 
and recommender system \cite{wang2015collaborative}; 
see 
\cite{goodfellow2016deep}
for an overview.
 Recent works have also demonstrated that deep learning can be applied in communication systems, either by replacing an individual function module in the system (such as signal detection \cite{ye2017power,sun2019deep},  channel decoding \cite{nachmani2018deep}, 
 channel estimation \cite{wen2018deep,sun2018limited}), or by jointly representing the entire system \cite{dorner2017deep,o2017introduction} for achieving state-of-the-art performance. 
Specifically, deep learning is a data-driven method in which a large amount of training data is used to train a deep neural network (DNN) for a specific task (such as power control). Once trained, such a DNN model will replace conventional algorithms to process data in real time. Existing works have shown that when the real-time data follows similar distribution as the training data, then such an approach can generate high-quality solutions for  non-trivial wireless tasks \cite{ye2017power, sun2019deep,wen2018deep, sun2018limited,sun2018learning,lee2018deep, liang2019towards,eisen2020optimal,shen2019lorm,huang2019fast,nachmani2018deep,cui2019spatial},
while significantly reducing real-time computation, and/or requiring only a subset of channel state information (CSI).  

{\bf Dynamic environment.}   However, it is often challenging to use these DNN based algorithms when the environment (such as CSI and user locations) keeps changing. There are three main reasons. 

\noindent {\bf 1)} It is well-known that naive DL based methods typically suffer from severe performance deterioration when the environment changes, that is, when the real-time data follows a different distribution than those used in the training stage \cite{shen2019lorm}.

\noindent {\bf 2)}  One can adopt the {\it transfer learning} and/or {\it online learning} paradigm, by updating the DNN model according to data generated from the new environment \cite{shen2019lorm}.  However, these approaches usually degrade or even overwrite the previously learned models \cite{parisi2019continual,mccloskey1989catastrophic}. Therefore  they are sensitive to outlier because once adapted to a transient (outlier) environment/task, its performance on the  existing environment/task can degrade significantly \cite{kirkpatrick2017overcoming}.  Such kinds of behavior are particularly undesirable for wireless resource allocation tasks, because the unstable model performance would cause large outage probability for communication users.

\noindent {\bf 3)} If the entire DNN is periodically retrained using all the data seen so far \cite{kirkpatrick2017overcoming}, then the training can be {time and memory} consuming since the number of data needed keeps growing.

Due to these challenges, it is unclear how   state-of-the-art DNN based communication algorithms could properly adapt to new environments quickly without {experiencing} significant performance loss over previously encountered environments. 
Ideally, one would like to design  data-driven models that can adapt to the new environment {\it efficiently} (i.e., by using as little resource as possible), {\it seamlessly} (i.e., without knowing when the environment has been changed), {\it quickly} (i.e., adapt well using only a small amount of data), and {\it continually} (i.e., without forgetting the  previously learned models). 

{\bf Continual  Learning.} In the machine learning community, continual learning (CL)  has recently been proposed to address the ``catastrophic forgetting phenomenon". {That is, the tendency of abruptly losing the previously learned models when the current environment information is incorporated \cite{mccloskey1989catastrophic}.} Specifically, consider the setting where different ``tasks" (e.g., different CSI distributions) are revealed sequentially. Then CL aims to retain the knowledge learned from the early tasks through  one of the following mechanisms: 1) regularize the most important parameters 
\cite{kirkpatrick2017overcoming,li2017learning};
2) design dynamic neural network architectures and associate neurons with tasks \cite{ yoon2018lifelong,zenke2017continual,rusu2016progressive}; 
or 3) introduce a small set of memory for later training rehearsal \cite{lopez2017gradient,shin2017continual,rebuffi2017icarl}. 
{However, most of the above mentioned methods} require the knowledge of the {\it task boundaries}, that is, the time stamp where an old task terminates and a new task begins. Unfortunately, such a setting does not suit wireless communication problems well, since the wireless environment  usually changes  continuously, without a precise changing point.  Only limited recent CL works have focused on boundary-free environments \cite{isele2018selective, aljundi2019gradient, rolnick2019experience}, but they  all focus on proposing general-purpose tools without considering any problem-specific structures. 
Therefore, it is unclear whether they will be effective in wireless communication tasks. 

\begin{figure}
  \centering
  \includegraphics[width=0.8\linewidth]{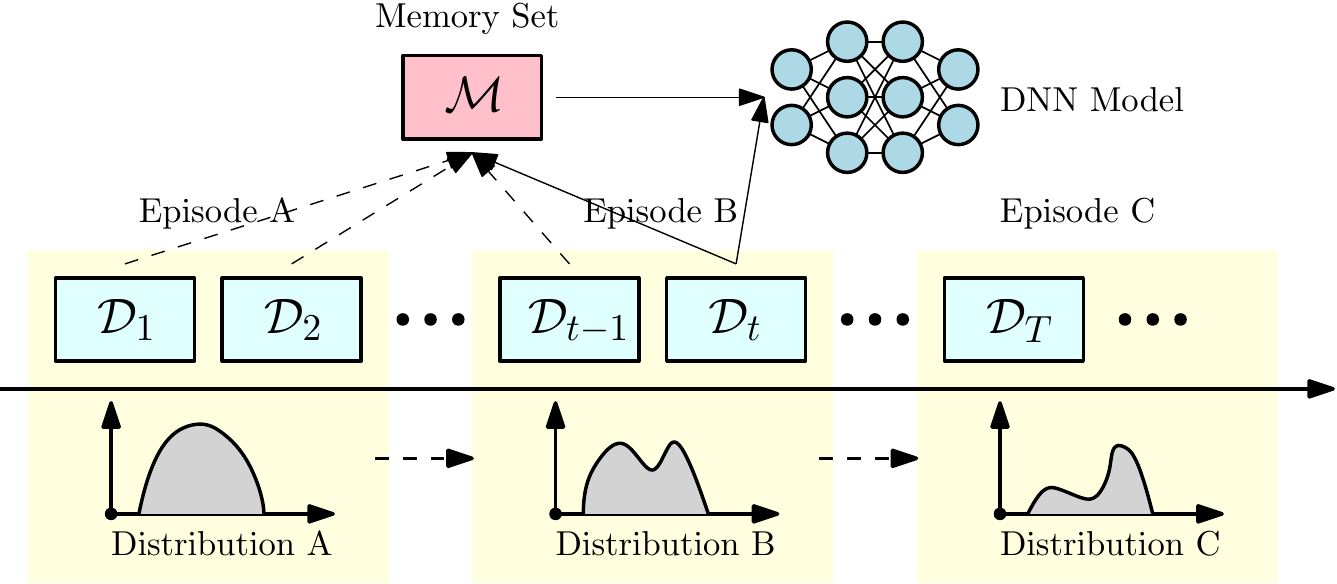}
  \caption{Proposed CL framework for the episodically dynamic environment. The data is feeding in a sequential manner (thus the system can only access $D_t$ at time $t$) with changing episodes and distributions, and the model has a limited memory set $\mathcal{M}$ (which cannot store all data $D_1$ to $D_t$). To maintain the good performance over all experienced data from $D_1$ to $D_t$, the proposed framework optimizes the data-driven model at each time $t$, based on the mixture of the current data $D_t$ and the memory set $\mathcal{M}$. The memory set $\mathcal{M}$ is then updated to incorporate the new data $D_t$. }
  \label{overview}
\end{figure}

{\bf Contributions.}
The {\it main contribution} of this paper is  that we introduce the notion of CL to data-driven wireless system design, and  develop a  tailored CL formulation together with a training algorithm. 
Specifically, we consider an ``episodically dynamic" setting where  the environment changes in {\it episodes}, and within each episode the distribution of the CSIs stays relatively stationary. 
Our goal is to design a learning model which can seamlessly and efficiently adapt to the changing environment, while maintaining the previously learned knowledge, and  without knowing the episode boundaries. 

Towards this end, we propose a CL framework for wireless systems, which incrementally adapts the DNN models by using data from the new episode as well as a limited but carefully selected subset of data from the previous episodes; see Fig. \ref{overview}.
Compared with the existing heuristic {boundary-free} CL  algorithms 
\cite{isele2018selective, aljundi2019gradient, rolnick2019experience}, our
approach is based upon a clearly defined optimization formulation that is tailored  for the wireless resource allocation problem. In particular, our CL method is based on a   bilevel  optimization which selects a small set of important data samples into the working memory according to certain {\it data-sample fairness}  criterion. We further relax the lower level of constrained non-convex bilevel problem using a smooth approximation, and propose and analyze practical (stochastic) algorithms for model training. Moreover, we demonstrate the effectiveness of our proposed framework by applying it to two popular DNN based models (one for power control and the other for beamforming). We test our CL approach using both synthetic and  ray-tracing  based data. To advocate reproducible research, the code of our implementation is available online at \url{https://github.com/Haoran-S/TSP_CL}.

 \section{Literature Review} 
\subsection{Deep learning for Wireless Communication}
Recently, DL has been used to generate high-quality solutions for non-trivial wireless communication tasks \cite{ye2017power, sun2019deep,wen2018deep, sun2018limited,sun2018learning,lee2018deep, liang2019towards,eisen2020optimal,shen2019lorm,huang2019fast,nachmani2018deep,cui2019spatial}.
These approaches can be roughly divided into following two categories: 
\subsubsection{End-to-end Learning}
For the classic resource allocation problems such as power control, the work \cite{sun2018learning} shows that DNNs can be exploited to learn the optimization algorithms such as WMMSE \cite{shi2011iteratively}, in an end-to-end fashion. Subsequent works such as \cite{lee2018deep} and \cite{liang2019towards} show that unsupervised learning can be used to further improve the model performance. Different network structures, such as convolutional neural networks \cite{lee2018deep} and graph neural networks \cite{shen2019graph, eisen2020optimal}, and different modeling techniques, such as reinforcement learning \cite{nasir2019multi}, are also studied in the literature. Nevertheless, all the above mentioned methods belong to the category of end-to-end learning, where a black-box model (typically deep neural network) is applied to  learn either the structure of some existing algorithms, or the optimal solution of a communication task.

\subsubsection{Deep Unfolding}
Alternatively, deep unfolding based methods \cite{balatsoukas2019deep} unfold existing optimization algorithms iteration by iteration and approximate the per-iteration behavior by one layer of the neural  network. 
In the machine learning community, well-known works in this direction include the unfolding of the iterative soft-thresholding algorithm (ISTA) \cite{gregor2010learning}, unfolding of the non-negative matrix factorization methods \cite{hershey2014deep}, and the unfolding of the alternating direction method of multipliers (ADMM) \cite{sprechmann2013supervised}. 
Recently, the idea of unfolding has been used in communication task  such as MIMO detection \cite{samuel2017deep, samuel2019learning, he2018model},  channel coding \cite{cammerer2017scaling}, 
resource allocation \cite{eisen2019learning}, channel estimation \cite{borgerding2017amp}, and beamforming problems \cite{hu2020iterative}; see a recent survey paper \cite{balatsoukas2019deep}.

\subsection{Continual Learning  }
\label{section:cl}
CL is originally proposed to improve reinforcement learning tasks \cite{ring1994continual} 
to help alleviate the catastrophic forgetting phenomenon, that is, the tendency of abruptly losing the knowledge about the previously learned task(s) when the current task information is incorporated \cite{mccloskey1989catastrophic}. It has later been broadly used to improve other machine learning models, and specifically the DNN models \cite{kirkpatrick2017overcoming,parisi2019continual}. 
Generally speaking, the CL paradigm can be classified into  the following categories.

\subsubsection{Regularization Based Methods} 
Based on the Bayesian theory and inspired by synaptic consolidation in Neuroscience, the regularization based methods penalize the most important parameters to retain the performance on old tasks \cite{kirkpatrick2017overcoming}.  Some most popular regularization approaches include Elastic Weight Consolidation (EWC) \cite{kirkpatrick2017overcoming} and Learning without Forgetting (LwF) \cite{li2017learning}.
 However, regularization or penalty based methods naturally introduce  tradeoff between the performance of old and new tasks. 
 If a large penalty is applied to prevent the model parameters from moving out of the optimal region of old tasks, 
 the model may be hard to adapt to new tasks; if a small penalty is applied, it may not be sufficient to force the parameters to stay in the optimal region to retain the performance on old tasks. 

\subsubsection{Architectures Based Methods} By associating neurons with tasks (either explicitly or not), many different types of dynamic neural network architectures are proposed to address the catastrophic forgetting phenomenon \cite{yoon2018lifelong}. However, due to the nature of the parameter isolation, architecture based methods usually require the knowledge of the task boundaries, and thus they are not suitable for wireless settings, where the environment change is often difficult to track. 

\subsubsection{Memory Based Methods} \label{replay-cat} 
Tracing  back  to the 1990s,  the  memory (aka. rehearsal) based  methods  play an  important  role  in  areas  such  as  reinforcement  learning \cite{robins1995catastrophic}.
As its name suggests, memory based methods store a small set of samples in memory for later training rehearsal, either through selecting and storing the most represented samples \cite{lopez2017gradient}
or use generative models to generate representative samples \cite{shin2017continual}. However, all above methods  require the knowledge of the task boundaries,  which are not suitable for wireless settings. Only recently, the authors of \cite{rolnick2019experience} proposed  boundary-free methods by selecting the samples through random reservoir sampling, which fills the memory set  with data that is  sampled from the streaming  episodes uniformly at random. More complex mechanisms are also introduced recently to further increase the {\it sampling diversity}, where the diversity is measured by either the samples' stochastic gradient directions \cite{aljundi2019gradient} or the samples' Euclidean distances \cite{isele2018selective}.

\subsection{Related methods} 
In this section, we discuss a few methods which also deal with streaming data, and  compare them with the CL approach. 

\subsubsection{Online Learning}
Online learning deals with the learning problems where the training data comes sequentially, and data distribution over time may or may not be consistent \cite{shalev2011online}.  
The ultimate goal of online learning is to minimize the cumulative loss over time, utilizing the previously learned knowledge.
In particular, when data sampling is independently and identically distributed, online gradient descent is essentially the stochastic gradient descent method, and all classic complexity results can be applied. On the other hand, when data sampling is non-stationary and drifts over time, online learning methods are more likely to adapt to the most recent data, at the cost of degrading the performance on past data \cite{aljundi2019online}. 

\subsubsection{Transfer Learning (TL)} 
Different from online learning, TL is designed to apply the knowledge gained from one task to another task, based on the assumption that related tasks will benefit each other  \cite{pan2009survey}. By transferring the learned knowledge from old tasks to new tasks, TL can quickly adapt to new tasks with fewer samples and less labeling effort. A typical application is the model fine-tuning on a (potentially small) user-specific problem (e.g., MNIST classification) based on some offline pre-trained model using a comprehensive dataset (e.g. ImageNet dataset). By applying the gained knowledge from the original dataset, the model can adapt to the new dataset  quickly with a few samples. Similar ideas have been applied in wireless settings recently \cite{shen2019lorm,yuan2020transfer,zappone2019model} to deal with scenarios that network parameters changes.
However, since the model is purely fine-tuned on the new dataset, after the knowledge transfer, the knowledge from the original model may be altered or overwritten, resulting in significant performance deterioration on the original problem \cite{kirkpatrick2017overcoming}.

\section{The Episodic Wireless Environment} \label{formulation}

The focus of this paper is to design learning algorithms in a {\it dynamic} wireless environment, so that the data-driven models we build can seamlessly, efficiently, and continually adapt to new environments. This section provides details about our considered {\it dynamic} environment, and discuss potential challenges. 

Specifically, we consider an ``episodically dynamic" setting where the environment changes relatively slowly in ``episodes", and during each episode the learners observe multiple {\it batches} of samples generated from {\it the same} stationary distribution; see Fig. \ref{overview}. We use $\mathcal{D}_t$ to denote a small batch of data collected at time $t$, and assume that each episode $k$ contains a set of $T_k$ batches, and use  $\mathcal{E}_k = \{\mathcal{D}_t\}_{t\in T_k}$ to denote the data collected in episode $k$. To have a better understanding about the setting, let us consider the following example. Again, we do not have knowledge about the episode boundaries. 

\subsection{A Motivating Example}
Suppose a collection of base stations (BSs) run certain DNN based resource allocation algorithm to provide coverage for a given area (e.g., a shopping mall). The users' activities can contain two types of  patterns: 1) regular but gradually changing patterns -- such as daily commute for the employees and customers, and such a kind of pattern could slowly change from week to week (e.g., the store that people like to visit in the summer is different in winter); 2) irregular but important patterns -- such as large events (e.g., promotion during the anniversary season), during which the distribution of user population (and thus the CSI distribution) will be significantly different compared with their usual distributions, and more careful resource allocation has to be performed.  {The episode, in this case, can be defined as ``a usual period of time", or ``an unusual period of time that includes a particular event".}

 For illustration purposes, suppose that each BS solves a weighted sum-rate (WSR) maximization problem for single-input single-output (SISO) interference channel, with a maximum of $K$ transmitter and receiver pairs.  Let $h_{kk}\in\mathbb{C}$ denote the direct channel between transmitter $k$ and receiver $k$, and $h_{kj}\in\mathbb{C}$ denote the interference channel from  transmitter~$j$ to receiver $k$.  
The power control problem aims to maximize the weighted system throughput via allocating each transmitter's transmit power $p_k$. For a given snapshot of the network, the problem can be formulated as the following:
\begin{align}\label{eq:sum_rate_IA}
\max_{p_1,\ldots, p_K}\quad &R(\bp; \bh):=\sum_{k=1}^K \alpha_k\log\left(1+\frac{|h_{kk}|^2p_k}{\sum_{j\neq k}|h_{kj}|^2p_j+\sigma_k^2}\right) \nonumber\\
\textrm{s.t.} \quad &0\leq p_k\leq P_{\max},~ \forall \; k=1,2,\ldots, K,
\end{align}
where $P_{\max}$ denotes the power budget of each transmitter; $\{\alpha_k>0\}$ are the weights. Problem \eqref{eq:sum_rate_IA} is known to be NP-hard \cite{luo2008dynamic} but can be effectively approximated by many optimization algorithms \cite{shi2011iteratively}. 
The data-driven methods proposed in recent works such as \cite{sun2018learning,lee2018deep, liang2019towards,shen2019lorm,eisen2020optimal} train DNNs using some pre-generated dataset. Here $\mathcal{D}_t$ can include a mini-batch of channels $\{h_{kj}\}$, and each episode can include a period of time where the channel distribution is stationary.

{For illustration purposes, let us consider the following scenario. At the beginning of a period, a DNN model for solving problem \eqref{eq:sum_rate_IA} (pretrained using historical data, $\mathcal{D}_0$) is preloaded on the BSs to capture the regular patterns in the shopping mall area. 
The question is, what should the BSs do  when the unexpected patterns appear? Say every morning a {\it morning model} is loaded to allocate resources up until noon. During this time the BSs can collect batches of data $\mathcal{D}_t, t=1,2, \cdots.$ Then shall the BS update its {\it morning model} immediately  to capture the dynamics of the user/demand distribution? If so, shall we use the entire data set, including the historical data and the real-time data, to re-train the neural network (which can be time-consuming), or shall we use TL to adapt to the new environment on the fly (which may result in overwriting the basic {\it morning model})? 
}

To address the above questions, we propose to adopt the notion of {CL}, so that our model can  incorporate the new data $\mathcal{D}_{t}$ on the fly, while keeping the knowledge acquired from $\mathcal{D}_{0:t-1}$. In the next section, we will detail our proposed CL formulation to achieve such a goal.  

\section{CL for Learning Wireless Resource} \label{CL_approach}

\subsection{Memory-based CL}
Our proposed method is based upon the notion of the {\it memory-based CL} proposed in \cite{lopez2017gradient,isele2018selective,aljundi2019gradient}, which allows the learner to collect a small subset of historical data for future re-training.  The idea is, once $\mathcal{D}_t$ is received, we  fill in  memory $\mathcal{M}_t$ (with fixed size) with the {\it most representative samples} from all experienced episodes $\mathcal{D}_{0:t-1}$, and then train the neural network at each time $t$ with the data $\mathcal{M}_t\cup\mathcal{D}_t$.  Several major features of this approach are listed below:

 \noindent $\bullet$ The learner does not need to know where a new episode starts (that is, the boundary-free setting) -- it can keep updating $\mathcal{M}_t$ and keep training as data comes in. 
 
  \noindent $\bullet$  If one can control the size of the memory well, then the training complexity will be made much smaller than performing a complete training over the entire data set $\mathcal{D}_{0:t}$, and will be comparable with TL approach which uses $\mathcal{D}_t$.
  
  \noindent $\bullet$  If the size of a given data batch $\mathcal{D}_t$ is very small, the learner is unlikely to {\it overfit} because the memory size is kept as fixed during the entire training process. This makes the algorithm more robust than the TL technique.

As mentioned before, existing memory-based CL methods include the random reservoir sampling algorithm \cite{rolnick2019experience}, 
and sample diversity based methods \cite{aljundi2019gradient,hayes2019memory,isele2018selective}. However, these works have a number of drawbacks. 
 First, for the reservoir sampling, if certain episode only contains a very small number of samples, then samples from this episode will be poorly represented in $\mathcal{M}_t$ because the probability of having samples from an episode in the memory is only dependent on the size of the episode.
Second, for the {diversity based methods}, the approach is again heuristic, since it is not clear how the ``diversity"
measured by large gradient {or Euclidean distances} can be directly linked to the quality of representation of the dataset. Third, and perhaps most importantly, the ways that the memory sets are selected are {\it independent} of the actual learning tasks at hand. The last property makes these algorithms broadly applicable to different kinds of learning tasks, but also prevents them from exploring application-specific structures. It is not clear whether, and how well these approaches will work for the wireless communication applications of interest in this paper.

\subsection{The Proposed Approach} \label{section-proposed}
In this work, we propose a new memory-based CL formulation that is tailored to the wireless resource allocation problem. 
{Our approach differs from the existing memory-based CL approaches discussed in the previous subsection,  because we intend to directly use features of the learning problem at hand to build our memory selection mechanism.}

To explain the proposed approach, let us begin with  presenting two common ways of formulating the training problem for learning optimal wireless resource allocation.
	First, one can adopt an {\it unsupervised learning} approach, which directly optimizes some native measures of wireless system performance, such as the throughput, Quality of Service, or the user fairness \cite{mo2000,hong12survey}, and this approach does not need any labeled data.
	Specifically,  a popular DNN training problem  is given by
\begin{align}\label{eq:unsupervised}
    \min_{\bTheta}  
    \sum_{i\in\mathcal{D}_{0:T}} \ell(\bTheta; \bh^{(i)}),
\end{align}
where $\bh^{(i)}$ is the $i$th CSI sample; $\bTheta$ is the DNN weight to be optimized;
$\ell(\cdot)$ is the negative of the per-sample sum-rate function, that is:
$\ell(\bTheta; \bh^{(i)}) = -R(\pi(\bTheta;\bh^{(i)}); \bh^{(i)})$,  where $R$ is defined in \eqref{eq:sum_rate_IA} and 
$\pi(\bTheta;\mathbf{h}^{(i)})$ is the output of DNN which predicts the power allocation. The advantage of this class of unsupervised learning approach is that the system performance measure is directly built into the learning model, while the downside is that this approach can get stuck at low-quality local solutions due to the non-convex nature of DNN \cite{song2021}.  

Secondly, it is also possible to use a supervised learning approach. Towards this end, we can generate some {\it labeled data} by executing a state-of-the-art optimization algorithm over all the training data samples  \cite{sun2018learning}. Specifically, for each CSI vector $\bh^{(i)}$, we can use algorithms such as the WMMSE \cite{shi2011iteratively} to solve problem \eqref{eq:sum_rate_IA} and obtain a high-quality solution $\bp^{(i)}$. Putting the $\bh^{(i)}$ and $\bp^{(i)}$ together yields the $i$th labeled data sample.
{Specifically,  a popular supervised DNN training problem is given by
\begin{align} 
    \min_{\bTheta}  
    \sum_{i\in\mathcal{D}_{0:T}} \ell(\bTheta; \bh^{(i)}, \bp^{(i)}),
\end{align}
where $\ell(\cdot)$ can be the Mean Squared Error (MSE) loss, that is:
$\ell(\bTheta; \bx^{(i)}, \bp^{(i)})   =  \|\bp^{(i)} - \pi(\bTheta, \bx^{(i)})\|^2$.}
Such a supervised learning approach typically finds high-quality models \cite{song2021,sun2018learning}, but  often incurs significant computation cost since generating high-quality labels can be very time-consuming. Additionally, the quality of the learning model is usually limited by that of label-generating optimization algorithms.

  Our idea is to leverage the advantages of both training approaches to construct a memory-based CL formulation. Specifically, we propose to select the most representative data samples $\bh^{(i)}$'s  into the working memory, by using a  {\it sample fairness} criteria. That is, those data samples that have relatively low system performance are more likely to be selected into the memory. Meanwhile, the DNN is trained by performing  either supervised or unsupervised learning over the selected data samples.  We expect that as long as the learning model can perform well on these challenging and {\it under-performing} data samples, then it should work well for the rest of the samples in a given episode.

To proceed, 
let us first assume that the entire dataset $\mathcal{D}_{0:T}$ is available. Let us use $\ell(\cdot)$ to denote a function measuring the per-sample training loss, $u(\cdot)$  a loss function measuring system performance for one data sample, $\bTheta$  the weights to be trained, $\bx^{(i)}$  the $i$th data sample and $\bp^{(i)}$ the $i$th label. Let $\pi(\bTheta;\bx^{(i)})$ denote the output of the neural network. Let us consider the following bilevel optimization problem   
 \begin{subequations}\label{eq:bilevel}
\begin{align}
\min_{\bTheta}  & \quad \sum_{i \in \mathcal{D}_{0:T}} \lambda^{(i)}_*(\bTheta) \cdot \ell (\bTheta; \bx^{(i)}, \bp^{(i)})  \label{eq:bilevel-a}\\
\text{\rm s.t.} & \quad\bblambda_*(\bTheta) = \arg \max_{\bblambda \in \mathcal{B}} \sum_{i \in \mathcal{D}_{0:T}} \lambda^{(i)} \cdot u(\bTheta; \bx^{(i)}, \bp^{(i)}), \label{eq:bilevel-b}
\end{align}
\end{subequations}
where $\mathcal{B}$ denotes the simplex constraint
$$\mathcal{B}:= \left\{\bblambda \quad \bigg|  \sum_{i\in \mathcal{D}_{0:T}} \lambda^{(i)} = 1, \quad \lambda^{(i)} \ge 0, \quad  \forall~i\in \mathcal{D}_{0:T}\right\}.$$

In the above formulation, the upper level problem \eqref{eq:bilevel-a} optimizes the {\it weighted} training performance across all data samples, and the lower level problem \eqref{eq:bilevel-b} assigns larger weights to those data samples that have higher loss $u(\cdot)$ (or equivalently, lower system level performance). 
The lower level problem has a linear objective, so the optimal $\bblambda_*$ is always on the {vertex} of the simplex, and the non-zero elements in $\bblambda_*$ all have the same weight. 
Such a solution naturally selects a subset of data for the upper level training problem to optimize. 

\begin{remark} \label{remark1}
{\bf (Choices of Loss Functions)} One feature of the above formulation is that we decompose the training problem and the  data selection problem, so that we can have the flexibility of choosing different loss functions according to the applications at hand.  Below we discuss a few alternatives.

First, the upper layer problem trains the DNN parameters $\bTheta$, so we can adopt any existing training formulation we discussed above. For example, if  supervised learning is used, then one common training loss is the MSE loss:
\begin{align} \label{MSE-loss}
\ell_{\rm MSE}(\bTheta; \bx^{(i)}, \bp^{(i)})  & =  \|\bp^{(i)} - \pi(\bTheta, \bx^{(i)})\|^2.
\end{align}

Second, the lower level loss function $u(\cdot)$ can be chosen as some adaptive weighted negative sum-rate for the $i$th data sample, which is directly related to system performance
\begin{align} \label{sumrate-loss}
   u(\bTheta, \bx^{(i)}, \bp^{(i)})  =& -  \alpha_i(\bTheta;\bx^{(i)}, \bp^{(i)})\cdot R(\pi(\bTheta, \bx^{(i)}); \bx^{(i)}).
\end{align} 
If we choose $\alpha_i(\bTheta;\bx^{(i)}, \bp^{(i)})\equiv 1, \; \forall~i$ , then the channel realization that achieves the worst throughput by the current DNN model will always be selected,  and the subsequent training problem will try to improve such ``worst case" performance. 
 Alternatively, when the achievable rates at samples across different episodes vary significantly (e.g., some episodes can have strong interference), then it is likely that the previous scheme will select data  only from a few episodes.  
Alternatively, we can choose $ \alpha_i(\bTheta;\bx^{(i)}, \bp^{(i)})= {1}/{\bar{R}(\bx^{(i)})}$, where $\bar{R}(\bx^{(i)})$ is the rate achievable by running some existing optimization algorithm on the sample $\bx^{(i)}$. This way, the data samples that achieve the worst sum-rate ``{relative}" to the state-of-the-art optimization algorithm is more likely to be selected. 
Empirically the ratio $R(\pi(\bTheta, \bx^{(i)}); \bx^{(i)}) / {\bar{R}(\bx^{(i)})}$ should be quite uniform  across data samples \cite{sun2018learning}, {so if there is one sample whose ratio is significantly lower than the rest, then we consider it as ``underperforming" and select it into the memory.} 
\end{remark}

\begin{remark} \label{minmax-algorithm}
{\bf (Special Case)} As a special case of problem \eqref{eq:bilevel},  one can choose $\ell(\cdot)$ to be the same as  $u(\cdot)$. 
Then the bilevel problem reduces to the following minimax problem, which optimizes the {\it worst case} performance (measured by the loss $\ell(\cdot))$) across all samples:
	\begin{align}\label{eq:minimax}
\min_{\bTheta} \max_{\bblambda\in \mathcal{B}}  & \quad \sum_{i \in \mathcal{D}_{0:T}} \lambda^{(i)} \cdot \ell (\bTheta; \bx^{(i)}, \bp^{(i)}).
\end{align}
When $\ell(\cdot)$ is taken as the negative per-sample sum-rate defined in \eqref{eq:unsupervised}, problem \eqref{eq:minimax}  is related to the classical minimax resource allocation  \cite{mo2000,Bengtsson99optimaldownlink,Razaviyayn12maxmin}, with the key difference that it does not achieve {fairness} {\it across users}, but rather to achieve {fairness} across {\it data samples}.  

Compared to the original bilevel formulation \eqref{eq:bilevel}, the minimax formulation \eqref{eq:minimax} is more restrictive but its properties have been relatively better understood. Many recent works have been developed for solving this problem, such as the two-time-scale Gradient Descent Ascent (GDA)  algorithm \cite{lin2020gradient}; see \cite{razaviyayn2020nonconvex} for a recent survey about related algorithms.

\end{remark}

At this point, {neither the bilevel problem \eqref{eq:bilevel} nor the minimax} formulation \eqref{eq:minimax} can be used to design CL strategy yet, because solving these problems requires the full data $\mathcal{D}_{0:T}$. To make these formulations useful for the considered CL setting, we make the following approximation.  Suppose that at $t$-th time instance, we have the memory  $\mathcal{M}_t$ and the new data set $\mathcal{D}_t$ available. Then, we propose to solve the following problem to select data points at time $t$:  
\begin{align}\label{bilevel-ep}
\min_{\bTheta}  & \quad \sum_{i \in \mathcal{M}_t \cup \mathcal{D}_t} \lambda_t^{(i)}(\bTheta) \cdot \ell (\bTheta; \bx^{(i)}, \bp^{(i)})   \\
\text{\rm s.t. } &\quad \bblambda_t(\bTheta) = \arg \max_{\bblambda \in \mathcal{B}_t} \sum_{i \in \mathcal{M}_t \cup \mathcal{D}_t} \lambda^{(i)} \cdot u(\bTheta; \bx^{(i)}, \bp^{(i)}),  \nonumber
\end{align}
where $\mathcal{B}_t$ denotes the simplex constraint
$$\mathcal{B}_t:= \left\{\bblambda \quad \bigg|  \sum_{i\in \mathcal{M}_t \cup \mathcal{D}_t} \lambda^{(i)} = 1,\quad \lambda^{(i)} \ge 0, \; \forall~i\in \mathcal{M}_t \cup \mathcal{D}_t  \right\}.$$
 More specifically, at a given time $t$, we will  collect $M$ data points $j\in \mathcal{M}_t \cup \mathcal{D}_t$ whose corresponding $\lambda^{(j)}$'s are the largest. These data points will form the next memory $\mathcal{M}_{t+1}$, and problem \eqref{bilevel-ep} will be solved again. The entire procedure will be shown shortly in Algorithm \ref{algo_PGD2}. 
 
\subsection{Reformulation}
In the previous section, we have proposed the CL framework and its optimization formulations. In this section, we will propose practical (stochastic) algorithms to solve those problems, and provide some basic analysis.

In general, at each time $t$, the non-convex bilevel problem \eqref{bilevel-ep} is very challenging to solve. Recent works on bilevel problems typically focus on solving problems with unconstrained and strongly convex inner problems \cite{hong2020two}. However, there is no generic theoretical guarantee available  when the outer problem is  non-convex, and the inner problem is constrained. In this section, instead of directly solving the bilevel problem \eqref{bilevel-ep}, we relax the original non-convex constrained lower level problem using {a softmax function \cite{goodfellow2016deep}, which is  a {\it smooth} approximation of the argmax function}:
\begin{align} \label{bilevel-closeform}
\min_{\bTheta}  &   \sum_{i \in \mathcal{M}_t \cup \mathcal{D}_t} \lambda^{(i)}_*(\bTheta) \cdot \ell (\bTheta; \bx^{(i)}, \bp^{(i)})  \\
\text{\rm s.t.} & \quad\lambda^{(i)}_*(\bTheta) = \frac{e^{u(\bTheta; \bx^{(i)}, \bp^{(i)})}}{\sum_{j \in \mathcal{M}_t \cup \mathcal{D}_t} e^{u(\bTheta; \bx^{(j)}, \bp^{(j)})}} \in (0, 1), \quad \forall i \nonumber.
\end{align}

After using the above approximation, $\bblambda$ is now implicitly constrained and can be computed in a closed-form. It is clear that the obtained $\bblambda_{*}(\bTheta)$ still allocates larger weights to larger loss values $u(\cdot)$. Further, we no longer need to solve two problems simultaneously, since we can easily obtain a single level problem by plugging the lower level problem into the upper problem.

Formally, at a given time $t$, problem \eqref{bilevel-closeform} can be written as the following  {\it compositional} optimization form: 
\begin{align}\label{opt0-2}
	\min_{\bTheta}~F^t(\bTheta)=\df(\dg(\bTheta); \bTheta),
\end{align}
where we have defined:
\begin{subequations} \small
\label{eq-relation} 
\begin{align}
\df(\bz;\bTheta)&:= \frac{\sum_{i \in \mathcal{M}_t \cup \mathcal{D}_t} e^{u(\bTheta; \bx^{(i)}, \bp^{(i)})} \cdot \ell(\bTheta; \bx^{(i)}, \bp^{(i)})}{|\mathcal{M}_t \cup \mathcal{D}_t| \cdot \bz},\\
\dg(\bTheta)&:= \frac{1}{|\mathcal{M}_t \cup \mathcal{D}_t|} \cdot \sum_{i \in \mathcal{M}_t \cup \mathcal{D}_t} e^{u(\bTheta; \bx^{(i)}, \bp^{(i)})}.
\end{align}
\end{subequations}

\subsection{Optimization Algorithms and  Convergence}
In this subsection, we will design algorithms for solving problem \eqref{opt0-2} (for a given time instance $t$). We first make the following standard assumptions.
\begin{assumption}[Boundedness] \label{ass_1}
The function value, the gradient and the Hessian of both upper level loss function $\ell(\cdot)$ and lower level loss function $u(\cdot)$ are bounded for all $\bTheta$, and for all $\bp\in[0, \mathbf{1}\times P_{\max}]$, and all realizations of $\bh$:
\begin{align*}
    \left\| \ell(\bTheta; \bx, \bp) \right\| &\le C_{\ell_0}, \quad  \left\|
      u(\bTheta; \bx, \bp) \right\|  \le C_{u_0},  \\
     \left\|\nabla_{\bTheta} \ell(\bTheta; \bx, \bp) \right\| & \le C_{\ell_1}, \quad 
    \left\|\nabla_{\bTheta} u(\bTheta; \bx, \bp) \right\| \le C_{u_1},   \\
    \left\|\nabla_{\bTheta}^2 \ell(\bTheta; \bx, \bp)\right\| &\le C_{\ell_2}, \quad 
    \left\|\nabla_{\bTheta}^2 u(\bTheta; \bx, \bp)\right\| \le C_{u_2}.
\end{align*}
\end{assumption}

{\remark Assumption \ref{ass_1} is reasonable in our specific problems. We can show that it can be satisfied if we choose $\ell(\cdot)$ and $u(\cdot)$ as suggested in  \eqref{MSE-loss} and \eqref{sumrate-loss}, and  use a neural network $\pi(\cdot)$ that have bounded gradient and Hessian \cite{herrera2020estimating}. The details  verifying   Assumption \ref{ass_1} will be left to   the supplemental material \ref{verify_ass}.}

Since the compositional problem \eqref{opt0-2} is essentially a single level problem, we can update  $\bTheta$ using the conventional gradient descent (GD) algorithm:
\begin{align}\label{GD} 
\bTheta^{k+1}&=\bTheta^k-\alpha \nabla \dg(\bTheta^k) \nabla_1 \df( \dg(\bTheta^k);\bTheta^k) \\
&\ \ \ -\alpha \nabla_2 \df( \dg(\bTheta^k);\bTheta^k), \nonumber
\end{align}
where $\alpha$ is the stepsize, and the two gradients are defined as 
\begin{equation*} 
\nabla_1 \df(\ba, \cdot) := \frac{\partial \df(\ba, \cdot)}{\partial \ba}, \quad \nabla_2 \df(\cdot, \mathbf{b}) := \frac{\partial \df( \cdot,\mathbf{b})}{\partial \mathbf{b}},
\end{equation*} 
for all $\ba, \mathbf{b}$ of appropriate sizes.

We have the following convergence result.
\begin{theorem} \label{thm-det}
Suppose Assumption \ref{ass_1} hold,  then the GD update \eqref{GD} achieves the following convergence rate
$$\min_{0\le k \le K} \|\nabla F^t(\bTheta^k)\|^2 \le \frac{c_0\cdot \bar{L} \cdot (F^t(\bTheta^0) - F^{t,*})}{K+1},$$ where $c_0$ is some universal positive constant, $\bar{L}$ is the Lipschitz constant of function $\nabla F^t(\bTheta)$, and $F^{t,*}$ is the optimal value of $F^t(\cdot)$ as defined in \eqref{opt0-2}. 
\end{theorem}
\begin{proof}
From Assumption \ref{ass_1} we can conclude that the function $F^t$ has Lipschitz continuous gradient with constant $\bar{L}$,  where the proof and precise definition of $\bar{L}$ is relegated to Lemma \ref{lemma-lip} in the supplemental material \ref{sec-lemma}. Then the desired  result immediately follows from the classical gradient descent analysis on non-convex problems; see \cite[Section 1.2.3]{nesterov2003introductory}. 
\end{proof}

However, the above update needs to evaluate $\dg(\bTheta)$, $\nabla \dg(\bTheta)$ and $\nabla \df(\bz, \bTheta)$, and the evaluation of each term requires the entire dataset $\mathcal{M}_t \cup \mathcal{D}_t$. This practically means that we need to perform full GD to train a (potentially large) neural network, which is computationally expensive, and typically results in poor performance.

A more efficient solution is to perform a stochastic gradient descent (SGD) type update, which first samples a mini-batch of data, then computes stochastic gradients to update. To be specific, the algorithm  samples a subset of data $\xi$ and $\phi$ uniformaly randomly at each iteration from the dataset $\mathcal{M}_t \cup \mathcal{D}_t$. Then the sampled versions of  $\df(\bz;\bTheta)$ and  $\dg(\bTheta)$ are given by:
\begin{subequations} 
	\label{eq-relation-sample}
	\begin{align}
	f(\bz;\bTheta;\xi)&:=  \frac{\sum_{i \in \xi} e^{u(\bTheta; \bx^{(i)}, \bp^{(i)})} \cdot \ell(\bTheta; \bx^{(i)}, \bp^{(i)})}{|\xi| \cdot \bz},\\
	g(\bTheta;\phi)&:= \frac{1}{|\phi|} \cdot \sum_{i \in \phi} e^{u(\bTheta; \bx^{(i)}, \bp^{(i)})},
	\end{align}
\end{subequations}
 where the notations $|\phi|$ and $|\xi|$ denote the number of samples in the mini-batch $\phi$ and $\xi$, respectively.
It is common to assume that the sampling mechanism can obtain $\xi$ and $\phi$ randomly and independently, that is, the following unbiasedness property holds.

\begin{assumption}[Unbiased Sampling]\label{ass_unbiased}
The sampling oracle satisfies the following  relations,  where  $z$ is a deterministic variable
\begin{align*}
\mathbb{E}_t\left[g(\bTheta;\phi)\right] &=\dg(\bTheta), \nonumber\\
\mathbb{E}_t\left[\nabla g(\bTheta;\phi) \right] &=\nabla \dg(\bTheta), \nonumber\\
\mathbb{E}_t\left[ \nabla_1 f(\bz; \bTheta;\xi)\right] &=  \nabla_1 \df(\bz;\bTheta)\nonumber\\
\mathbb{E}_t\left[ \nabla_2 f(\bz; \bTheta;\xi)\right] &=  \nabla_2 \df(\bz;\bTheta).
\end{align*}
Note that we have used the simplified notation $\mathbb{E}_t[\cdot]$ to indicate that the expectation is taken over the sampling process from the data set $\mathcal{M}_t \cup \mathcal{D}_t$.
\end{assumption}

Based on the above assumption, problem \eqref{opt0-2} can be equivalently written as:
\begin{align}\label{eq.stochastic.formulation}
	\min_{\bTheta}~F^t(\bTheta)=\mathbb{E}_t[{f}(\mathbb{E}_t[{g}(\bTheta;\phi)]; \bTheta; \xi)].
\end{align}
Then we can write down the following stochastic update, where the update direction $\mathbf{d}^k$ is an unbiased estimator of $\nabla F^t(\bTheta^k)$: 
 \begin{align} \label{update:theta} 
	\bTheta^{k+1}& = \bTheta^k - \alpha \mathbf{d}^k,
\end{align}
where we have defined:
\begin{align*}
\mathbf{d}^k & :=  \nabla g(\bTheta^k;\phi^k) \nabla_1 f\left(\mathbb{E}_t\left[g(\bTheta^k;\phi^k)\right];\bTheta^k;\xi^k\right)  + \nabla_2 f\left(\mathbb{E}_t\left[g(\bTheta^k;\phi^k)\right];\bTheta^k;\xi^k\right).
\end{align*}
Unfortunately, computing ${\bf d}^k$ is still costly due to the need to evaluate $\mathbb{E}_t\left[g(\bTheta^k;\phi^k)\right]$ (i.e.,  evaluating $\dg(\bTheta^k)$), which still involves the full data. One can no longer directly replace $\mathbb{E}_t\left[g(\bTheta^k;\phi^k)\right]$ by its stochastic samples $g(\bTheta^k; \phi)$ because such an estimator is biased, that is: 
\begin{align*}
 \mathbb{E}_t[\nabla g(\bTheta^k;\phi^k) \nabla_1 f(g(\bTheta^k;\phi^k);\bTheta^k;\xi^k)] \ne & \nabla \dg(\bTheta^k) \nabla_1 \df(\dg(\bTheta^k);\bTheta^k).
\end{align*}
To proceed, we introduce an auxiliary sequence $\{\by^{k+1}\}$ to track  $\mathbb{E}_t[g(\bTheta^k;\phi^k)]$. The resulting SGD-type algorithm is given below: 
\begin{subequations}\label{eq.SCSC}
\begin{align}
	\bTheta^{k+1}&=\bTheta^k-\alpha_k \nabla g(\bTheta^k;\phi^k) \nabla_1 f(\by^{k+1};\bTheta^k;\xi^k) - \alpha_k \nabla_2 f(\by^{k+1};\bTheta^k;\xi^k), \label{eq.SCSC-1}\\
		\by^{k+1}&=(1-\beta_k)\left(\by^k+g(\bTheta^k;\phi^k)-g(\bTheta^{k-1};\phi^k)\right) +\beta_k g(\bTheta^k;\phi^k),\label{eq.SCSC-2}
\end{align}
\end{subequations} 
where $\{\alpha_k\}$ and $\{\beta_k\}$ are sequences of stepsizes. The rationale is that, if the auxiliary variable $\by^{k+1}$ can track the true value $\dg(\bTheta^k)$ reasonably well, then \eqref{eq.SCSC-1} will be able to approximate an unbiased estimator of the true gradient. 

 Finally, {the overall stochastic algorithm for approximately solving problem \eqref{eq:bilevel} is given in Algorithm \ref{algo_PGD2}.  For each time period $t$, we first solve the relaxed problem \eqref{bilevel-closeform} in line 4-8,  by performing the stochastic updates described in \eqref{eq.SCSC-1} -- \eqref{eq.SCSC-2} for $K$ times (where $K$ is a predetermined number). Next, we construct the memory set $\mathcal{M}_t$ in line 9-13.} We sort the elements of $\{\lambda_t^{(i)}\}$ (defined in \eqref{bilevel-closeform}) and pick $M$ largest elements' index set $\mathcal{I}$ (In line 12 of the table); Then we assign the data points associated with the index set $\mathcal{I}$ to the new memory set $\mathcal{M}_{t+1}$ (In line 13 of the table).

\begin{algorithm2e}[t]
	\small
	\caption{{The Proposed stochastic CL Algorithm}}\label{algo_PGD2}
	Input: Memory $\mathcal{M}_0=\emptyset$, memory size $M$, max iterations $R$, step-sizes $\alpha$, $\beta$\\
	\While {receive $\mathcal{D}_t$}{
		Set $\mathcal{G}_t = \mathcal{M}_{t}\cup\mathcal{D}_t$

		\For{$k = 1:K$}{
			$\bTheta^{k+1} \leftarrow \bTheta^k-\alpha_k \nabla g(\bTheta^k;\phi^k) \nabla_1 f(\by^{k+1};\bTheta^k;\xi^k)$ 
			
			$\quad  \quad \quad - \alpha_k \nabla_2 f(\by^{k+1};\bTheta^k;\xi^k)$
			
			$\by^{k+1} \leftarrow (1-\beta_k)\left(\by^k+g(\bTheta^k;\phi^k)-g(\bTheta^{k-1};\phi^k)\right)$
			
			$\quad  \quad \quad +\beta_k g(\bTheta^k;\phi^k)$  
		}
		\uIf{ $|\mathcal{G}_t|<M$}{
			$\mathcal{M}_{t+1}  = \mathcal{G}_t$ 
		}
		\Else{
			$\mathcal{I} = \mbox{\rm Top}_M (\{\lambda_t^{(i)}\}_{\forall i})$ 
			
			$\mathcal{M}_{t+1}  = \{\mathcal{G}_t^{(i)}\}_{ i\in \mathcal{I}}$ 
		}
	}
\end{algorithm2e}

{\remark Note that the use of the auxiliary variable $\by$, first appeared in solving stochastic compositional optimization problems in the form of 
	\begin{align}\label{eq.composite}
	\min_{\bTheta} \mathbb{E}_{\xi}[f(\mathbb{E}_{\phi}[g(\bTheta,\phi)],\xi)],
	\end{align}
	see recent works \cite{wang2017stochastic,chen2020solving}, and the references therein. In particular, the authors of \cite{chen2020solving} provided the exact update form of \eqref{eq.SCSC-2}, and showed that the resulting algorithm enjoys the same sample efficiency as directly applying the SGD to solve problem \eqref{eq.composite}. 
	
	However, problem \eqref{eq.composite} is not exactly the same as our problem \eqref{eq.stochastic.formulation} because our problem includes an extra variable $\bTheta$ in the definition of $f(\cdot)$, so the update \eqref{eq.SCSC-1}  includes an additional term $- \alpha_k \nabla_2 f(\by^{k+1};\bTheta^k;\xi^k)$.  
	Therefore, more refined analysis steps have to be taken compared to \cite{chen2020solving}.
}

Below, we analyze the convergence of the $\bTheta$ and $\by$ updates given in line 4-8 of Algorithm 1. The following lemma is an immediate consequence of Assumption \ref{ass_1} and \ref{ass_unbiased}. Its proof is similar to Lemma \ref{lemma-lip-stoc} in the supplementary material Sec. \ref{sec-lemma}.  
\begin{lemma} \label{ass_bv_remark} 
Suppose Assumption  \ref{ass_1} -- \ref{ass_unbiased} hold, then we have \\
(1) The stochastic function $g(\bTheta;\phi)$ has bounded variance, that is, there exits a positive constant $V_g$ such that:
\begin{equation*}
\mathbb{E}_t\left[\|g(\bTheta;\phi)-\dg(\bTheta)\|^2\right]\leq V_g, \; \forall~\bTheta,
\end{equation*}
where $\phi$ denotes the random  data sampled from $\mathcal{M}_t \cup \mathcal{D}_t$, $\dg$ and $g$ are defined in \eqref{eq-relation} and \eqref{eq-relation-sample}, respectively.

\noindent (2) The stochastic gradient of $g$ is bounded in expectation, that is, there exists a positive constant $C_g$ such that 
\begin{align}\label{eq.ass2-g}
\begin{split}
  \mathbb{E}_t\left[\|\nabla g(\bTheta;\phi)\|\right] &\leq C_g, \; \forall~\bTheta.
\end{split}
\end{align}
(3) Fixing any sample $\phi$, the stochastic gradient of $g$ is $L_g$-smooth, that is, for any $\bTheta, \bTheta'\in\mathbb{R}^d$, we have: 
\begin{align*}
\begin{split}
\|\nabla g(\bTheta;\phi)-\nabla g(\bTheta';\phi)\| & \leq L_g\|\bTheta-\bTheta'\|.
\end{split}
\end{align*}
\end{lemma}

Next, we show that the tracking error of the auxiliary variable $\by$ is shrinking.
\begin{lemma}[Tracking Error Contraction  {\cite[Lemma 1]{chen2020solving}}]
\label{lemma2}
Consider ${\cal F}^k$ as the collection of random variables, i.e., ${\cal F}^k:=\left\{\phi^0, \ldots, \phi^{k-1}, \xi^0, \ldots, \xi^{k-1}\right\}$. 
Suppose Assumption  \ref{ass_1} and \ref{ass_unbiased}  hold, and $\by^{k+1}$ is generated by running  iteration \eqref{eq.SCSC-2}  conditioned on ${\cal F}^k$. The mean square error of $\by^{k+1}$ satisfies
\begin{align*}\label{eq.lemma2}
&\mathbb{E}_t\left[\|\dg(\bTheta^k)-\by^{k+1}\|^2\mid{\cal F}^k\right]\\
\leq & (1-\beta_k)^2\|\dg(\bTheta^{k-1})-\by^k\|^2+4(1-\beta_k)^2C_g^2\|\bTheta^k-\bTheta^{k-1}\|^2+2\beta_k^2V_g^2,
\end{align*}
where $C_g$ and $V_g$ are defined in Lemma \ref{ass_bv_remark}.
\end{lemma}
\begin{proof}
The above analysis is the same as the one presented for solving stochastic compositional optimization problems in the form of \eqref{eq.composite}  (see \cite[Lemma 1]{chen2020solving}). We include the steps in    the supplemental material \ref{eq.lemma2-prof} for completeness. \end{proof}

 We are now ready to show our main results about the convergence  of the sequence $\{(\bTheta^k, \by^k)\}_{k=1}^{K}$ in Algorithm 1. 
\begin{theorem}[Convergence Analysis] \label{theorem1}
 Consider Algorithm  \ref{algo_PGD2}, and fix a time instance $t$. Let $K$ be the total number of iterations used at time $t$ to update the tuple $\{(\bTheta^k, \by^k)\}$ . Suppose Assumptions \ref{ass_1} and \ref{ass_unbiased}  hold, and that the sequence of the auxiliary variable $\{\by^k\}$ is bounded away from zero, i.e., $\|\by^k\| \ge C_y, \; \forall~k$, for some positive constant $C_y$. Let us choose the stepsizes as 
$\alpha_k  ={\beta_k}/{L_0}, \; \forall~k$,  for some appropriately chosen $L_0>0$ (defined in \eqref{eq:L0} in the Appendix).
Then the iterates $\{\bTheta^k\}$ generated by the algorithm satisfies:
\begin{equation*} 
    \frac{1}{K}\sum_{k=0}^{K-1}\mathbb{E}_t[\|\nabla F^t(\bTheta^k)\|^2] \leq\frac{2F^t(\bTheta^0)+2\tilde{C}}{\sqrt{K}},
\end{equation*}
where $\tilde{C}$ is some universal constant,  dependent on Assumption \ref{ass_1},  \ref{ass_unbiased} and $C_y$. 
\end{theorem}
\begin{proof}
The full proof is relegated to Appendix \ref{eq.theorem1-proof}, and    $\tilde{C}$ is defined in \eqref{ctilde}. 
\end{proof}
{\remark The key idea of the proposed method is to use an auxiliary variable $\by$ to track the expected value $\mathbb{E}_t[g(\bTheta;\phi)]$ (or equivalently $\dg(\bTheta)$). Lemma \ref{lemma2} shows that the tracking error $\|\by - \dg(\bTheta)\|$ is shrinking given that $\alpha$ and $\beta$ are small. Theorem \ref{theorem1} implies that, for a given time instance $t$, the sequence $\{(\bTheta^k, \by^k)\}_{k=1}^{K}$ converges in the order of ${\cal O}(K^{-\frac{1}{2}})$, which is  the same order achieved by generic SGD methods for  non-compositional non-convex problems.
}

 Note that compared with Theorem \ref{thm-det}, we have made an additional assumption that the size of the iterates $\{\by^k\}$  is bounded away from zero. Although such an assumption cannot be verified {\it a priori},  in our numerical result it appears to always hold. Intuitively, this assumption makes sense since $\by$ tracks $\dg(\cdot)$, and $\dg(\cdot)$ is bounded  away from zero by its definition \eqref{eq-relation}. Therefore, as long as the tracking error is small (cf. Lemma \ref{lemma2}), we can assume $\by^k$  to be bounded  away from zero. 

Nonetheless,  we would like to emphasize that, the main contribution of this work is the development of the CL formulation and approximation problem  \eqref{bilevel-closeform}, as well as a set of practical algorithms for solving them. The convergence analysis helps us justify our design principle, but ultimately the efficiency of the proposed formulation and algorithms has to be tested in practice.  This is what we plan to do in the next section.

\section{Experimental Results} \label{simulation}
In this section, we illustrate the performance of the proposed CL framework.  
We choose two applications 
where the end-to-end learning based DNN is used: 1) power control for weighted sum-rate (WSR) maximization problem \cite{sun2018learning} with single-input single-output (SISO) interference channel defined in  \eqref{eq:sum_rate_IA};
 2) coordinated beamforming problem for the millimeter wave system \cite{alkhateeb2018deep},  with up to 256 antennas per BS.

\subsection{Simulation Setup}
The experiments are conducted on Ubuntu 18.04 with Python 3.6, PyTorch 1.6.0, and MATLAB R2019b on one computer with two 8-core Intel Haswell processors and 128 GB of memory. The codes are made available online through \url{https://github.com/Haoran-S/TSP_CL}.

\subsection{Randomly Generated Channel}\label{subsection: Randomly Generated Channel}
We first demonstrate the performance of our proposed framework using randomly generated channels, for a scenario with $K=10$ transmitter-receiver pairs. We choose three standard types of random channels used in previous resource allocation literature \cite{sun2018learning,liang2019towards} stated as following:

\noindent{\bf Rayleigh fading:} Each channel coefficient $h_{ij}$ is generated according to a standard normal distribution, i.e.,  
\begin{align}\label{Rayleigh}
   {\sf Re}(h_{ij})\sim \frac{\mathcal{N}(0, 1)}{\sqrt{2}}, \; {\sf Im}(h_{ij})\sim \frac{\mathcal{N}(0, 1)}{\sqrt{2}},  \; \forall i, j \in \mathcal{K}. 
\end{align} 
\\\noindent{\bf Rician fading:} Each channel coefficient $h_{ij}$ is generated according to a Gaussian distribution with 0dB $K$-factor, i.e., 
  $${\sf Re}(h_{ij})\sim \frac{1+\mathcal{N}(0, 1)}{2}, \; {\sf Im}(h_{ij})\sim \frac{1+\mathcal{N}(0, 1)}{2}, \; \forall i, j \in \mathcal{K}.$$ 
\\\noindent{\bf Geometry channel:}
All transmitters and receivers are uniformly randomly distributed in a R$\times$R area. The channel gains $| h_{ij}|^2$  follow the pathloss function
\begin{align*}
    | h_{ij}|^2= \frac{1}{1+d_{ij}^2} | f_{ij}|^2, \; \forall~i,j,
\end{align*}
where $f_{ij}$ is the small-scale fading coefficient follows $\mathcal{CN}(0, 1)$, $d_{ij}$ is the distance between the $i$th transmitter and $j$th receiver.

Then we use these coefficients to generate four different episodes: the Rayleigh fading channel, the Ricean fading channel, and the geometry channel (with nodes distributed in a $10m \times 10m$ and a $50m \times 50 m$ area, respectively). {We use such drastically changing environments to simulate (perhaps overly harsh) ``toy" scenarios.}
Later we will utilize real data to generate more practical scenarios. For each episode, we generate $20,000$ channel realizations for training and $1,000$ for testing. We also stacked the test data from all episodes to form a mixture test set, i.e., containing $4,000$ channel realizations. During the training stage, a total of $80,000$ channel realizations are available. A batch of $5,000$ realizations is  revealed each time, and the memory space contains only $2,000$ samples from the past. That is, $|\mathcal{D}_{1:16}| = 80,000$, $|\mathcal{D}_{t}|=5,000, |\mathcal{M}_t|=2,000, \; \forall~t$. 

For the data-driven model, we use the end-to-end learning based fully connected neural network model as implemented in  \cite{sun2018learning}. For each data batch of $5,000$ realizations at time $t$, we train the model $\bTheta_t$ using the following six different approaches for $20$ epochs (with the previous  model $\bTheta_{t-1}$ as initialization):
\begin{enumerate}
    \item Transfer learning (``TL") \cite{shen2019lorm} -- update  the  model  using the current data batch (a total of $5,000$ samples);

    \item Reservoir sampling based CL (``Reservoir")  \cite{isele2018selective} -- update the model using both the current data batch and the memory set (a total of $7,000$ samples), where data samples in the memory set are uniformly randomly sampled from the streaming episodes;

    \item Proposed fairness based CL (``Bilevel") in Algorithm \ref{algo_PGD2}  -- update the model using both the current data batch and the memory set (a total of $7,000$ samples), where data samples in the memory set are selected   according to the proposed data-sample fairness criterion \eqref{bilevel-ep}  using  Algorithm \ref{algo_PGD2}. Unless otherwise specified, as suggested in Section \ref{section-proposed},  the training loss $\ell(\cdot)$ is chosen as the MSE loss \eqref{MSE-loss}, the  system performance loss $u(\cdot)$ is chosen as the adaptive  weighted  negative  sum-rate loss \eqref{sumrate-loss}, and the weights is chosen as the sum-rate  achievable  by  the WMMSE method \cite{shi2011iteratively}. 
    
    \item Proposed minimax based CL (``Minimax")  {-- this is the special case of Algorithm 1, as described in remark \ref{minmax-algorithm}; In particular, we} update the model using both the current data batch and the memory set (a total of $7,000$ samples), where data samples in the memory set are selected according to the proposed minimax criterion \eqref{eq:minimax},   the training loss $\ell(\cdot)$ and the system performance loss $u(\cdot)$ are chosen as the MSE loss \eqref{MSE-loss}; The model is trained using  the gradient descent ascent (GDA) \cite{lin2020gradient}.
    
     \item Joint training (``Joint (equal)") -- update the model using all accumulated data up to current time (up to $80,000$ samples); All data points are treated equally, that is, $\lambda^{(i)}$ in \eqref{eq:bilevel} are equal for all $i$, and there is no lower level problem.
    
    \item Joint training (``Joint (weighted)") -- update the model using all accumulated data up to current time (up to $80,000$ samples); The proposed fairness based formulation \eqref{bilevel-closeform} is applied but replace the training set $\mathcal{M}_t \cup \mathcal{D}_t$ with all accumulated data.

\end{enumerate}

The simulation results of six different approaches are compared and shown in Fig.  \ref{fig-random-rate}.   Specifically, each subplot of Fig. \ref{fig-random-rate} (a) shows the performance of the time-varying models trained by different approaches as training data is streaming in, while evaluated at test samples drawing from the  episode specified by each subplot. The grey lines indicate the transition points for two consecutive training episodes. The $x$-axis represents the number of  training data that has been seen by the model, while the $y$-axis  represents the sum-rate achieved on the test data. 
 Fig. \ref{fig-random-rate} (b) shows the average of all four subplots from Fig. \ref{fig-random-rate} (a). 
Note that the joint training method uses up to $80,000$ in memory spaces and thus violates our memory limitation (i.e., $7,000$ in total), and the transfer learning method adapts the model to new data each time and does not use any additional memory spaces. One can observe that the proposed CL based methods perform well over all tasks, nearly matching the performance of the joint training method, whereas the TL  suffers from some significant performance loss as the ``outlier" episode comes in (i.e., geometry channel in our case).

\begin{figure}
\centering
\begin{minipage}[c]{0.4\linewidth}
    \centering
    \includegraphics[width= \linewidth]{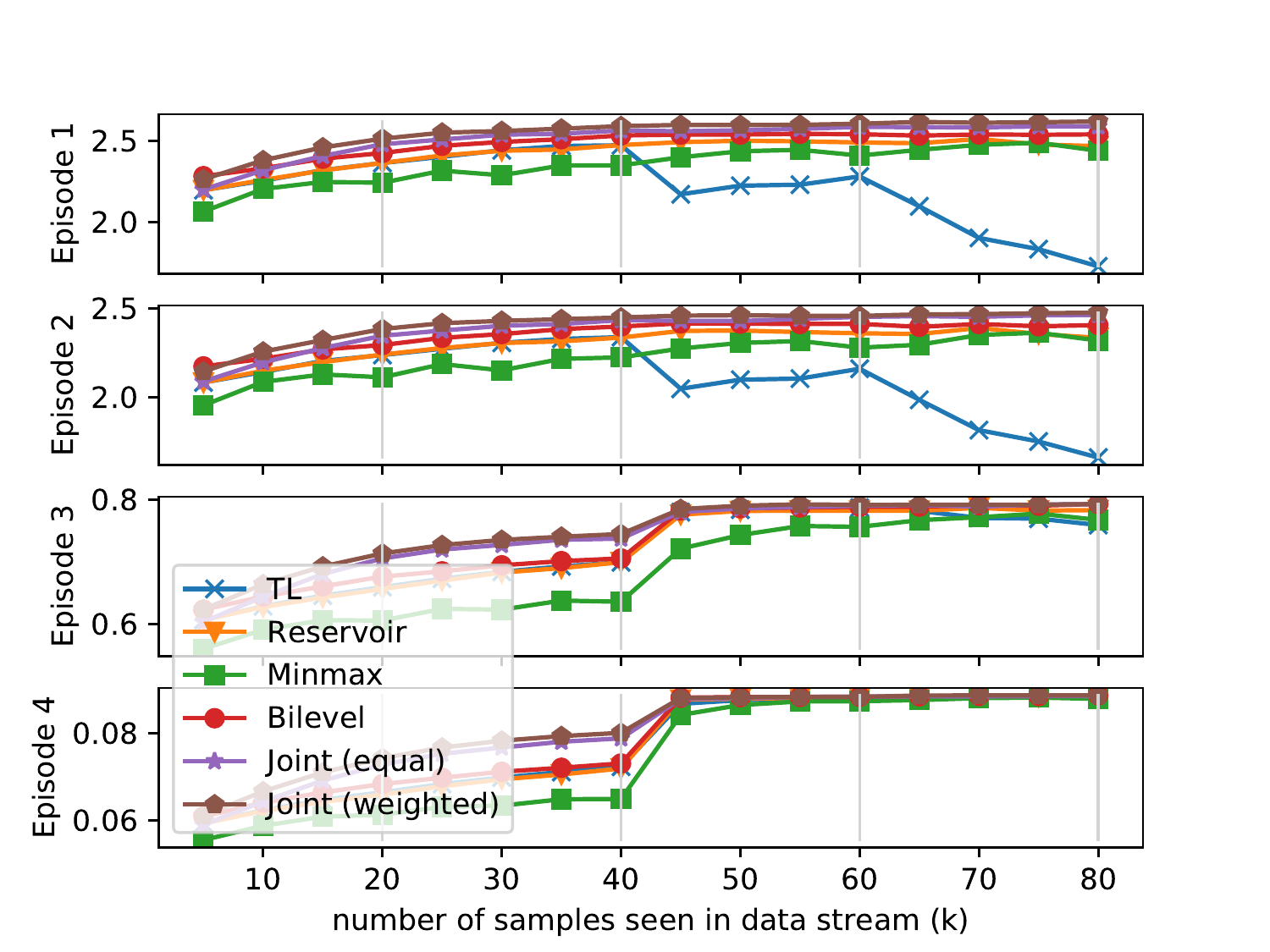}
    {\footnotesize (a) Test performance for each episode}
\end{minipage}
\noindent
\begin{minipage}[c]{0.4\linewidth}
    \centering
    \includegraphics[width= \linewidth]{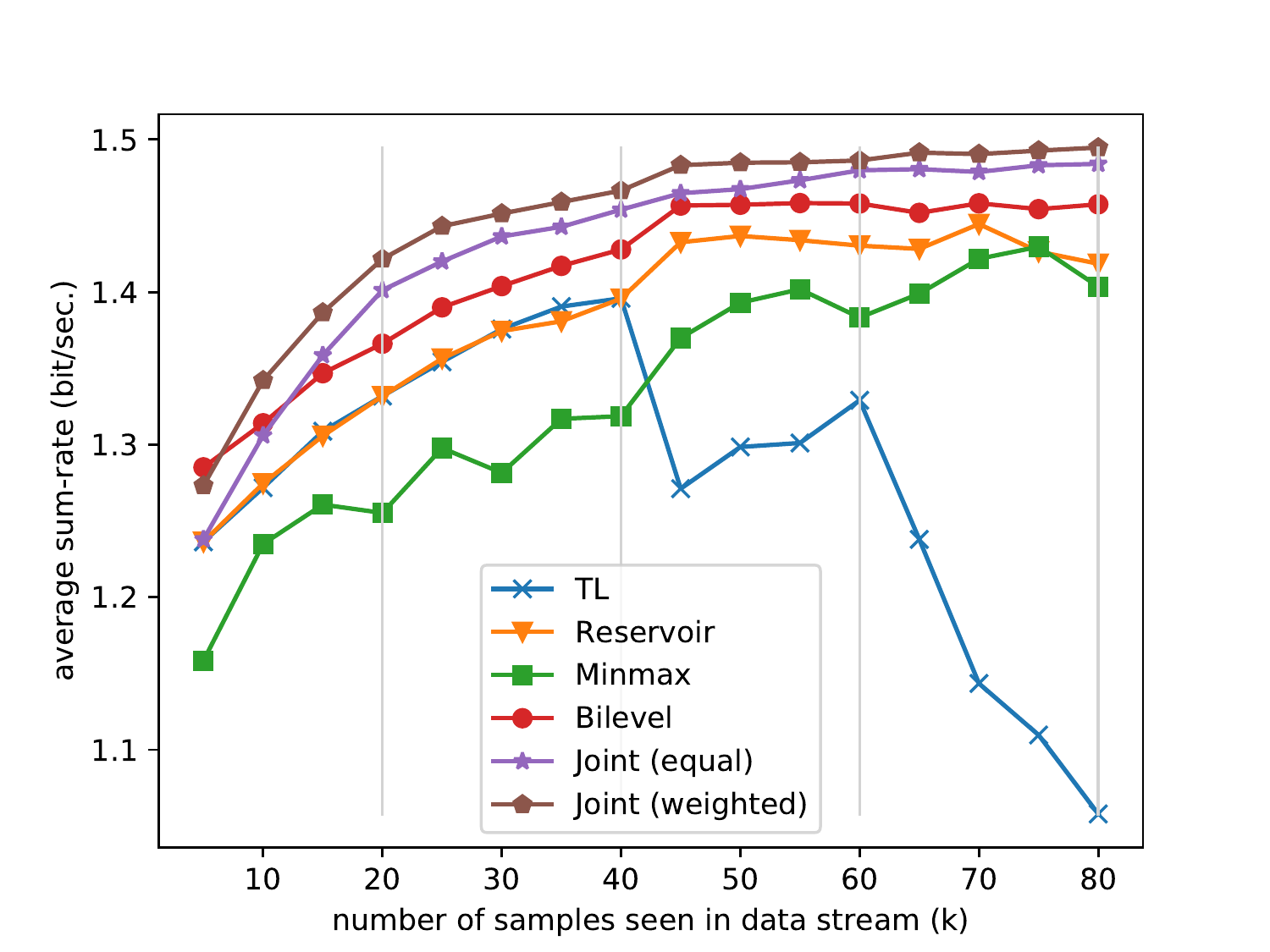}
    {\footnotesize (b) Average test performance for all episodes}
\end{minipage}
\caption{Testing sum-rate performance on randomly generated channels for (a) each individual episode and (b) average of all episodes. Each sub-figure of (a) represents the testing performance on the data generated from a particular episode (indicated in the y-axis). The grey line indicates the time instances where a new episode starts, which is unknown during training time. }
\label{fig-random-rate}
\end{figure}

\begin{figure}
\centering
\begin{minipage}[c]{0.5\linewidth}
    \centering
    \includegraphics[width=\textwidth]{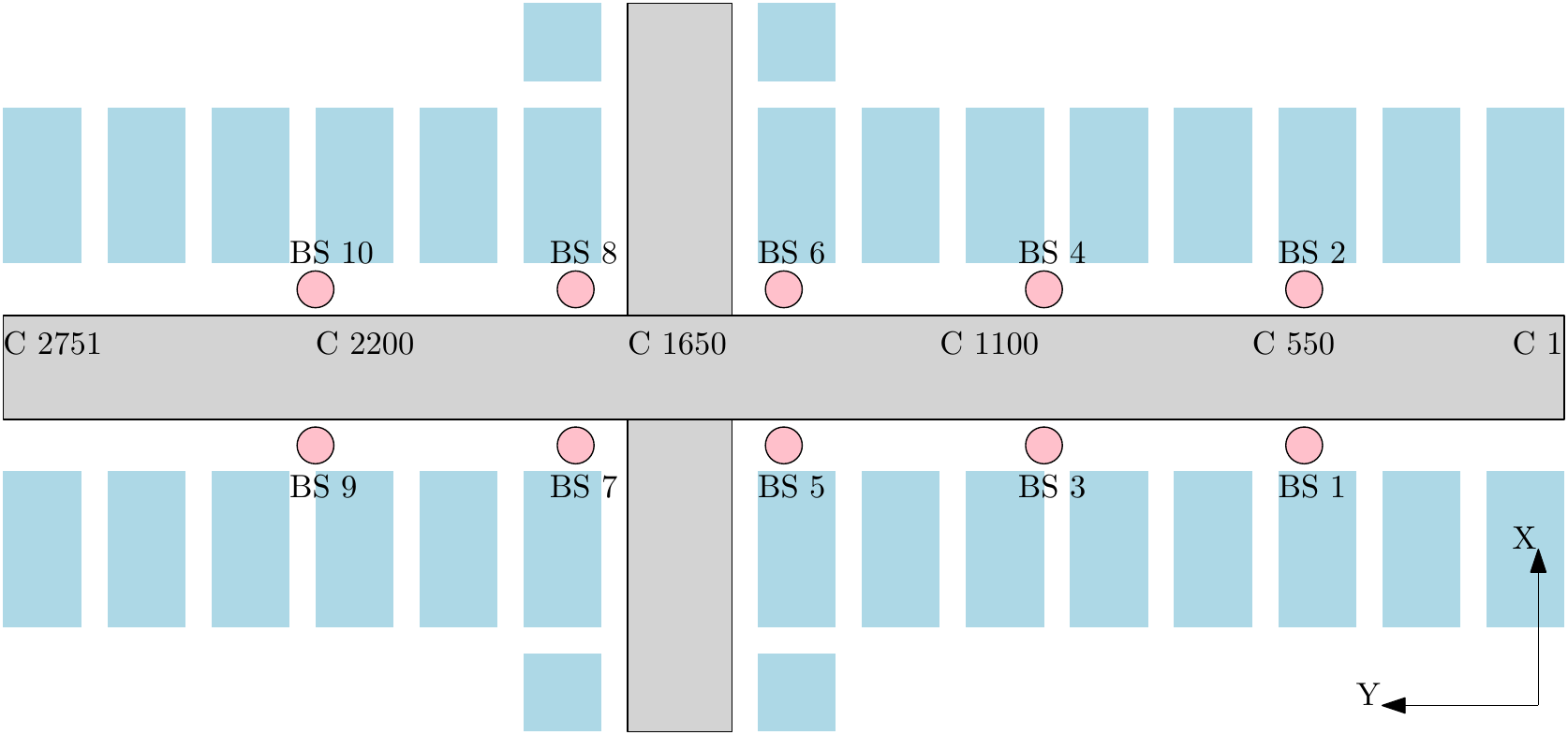}
\end{minipage}
    \caption{Top view of the DeepMIMO ray-tracing scenario \cite{Alkhateeb2019}, showing the two streets (grey rectangular), the buildings (blue rectangular), and the 10 base stations (red circle).}
    \label{location}
\end{figure}

\subsection{Real Measured Channel}

To validate our approach under more realistic scenarios, we further consider the outdoor `O1' ray-tracing scenario generated from the DeepMIMO dataset \cite{Alkhateeb2019}.
 The used dataset consists of two streets and one intersection, with the top-view showed in Fig. \ref{location}. 
The user grid is located along the horizontal street, with a length of 550m and a height of 36m. This street is divided into a $181 \times 2751$  grid, and the users could be located on any grid point. We index the first column from the right as $C1$, and the last column from the left as $C2751$. 
	
An episode is generated by using a particular user distribution. More specifically, the users for episode 1 are all drawn from columns C551 - C1100, and similarly, users from episode 2-3 are from C1101 - C1650, and C1651 - C2200, respectively. For each episode, we generate $20,000$ channel realizations for training and $1,000$ for testing. For each channel realization, we generate the channel based on 10 BSs (i.e., red circles in Fig. \ref{location}), and randomly pick $K = 10$ user locations from the selected user population. The BS is equipped with single antenna and has the maximum transmit power $p_k=30$dBm. The noise power is set to $-80$ dBm.  
 
\subsubsection{Average sum-rate comparison}
\label{sec:sum-rate-compare}

\begin{figure}
\centering
\begin{minipage}[c]{0.4\linewidth}
    \centering
    \includegraphics[width=\textwidth]{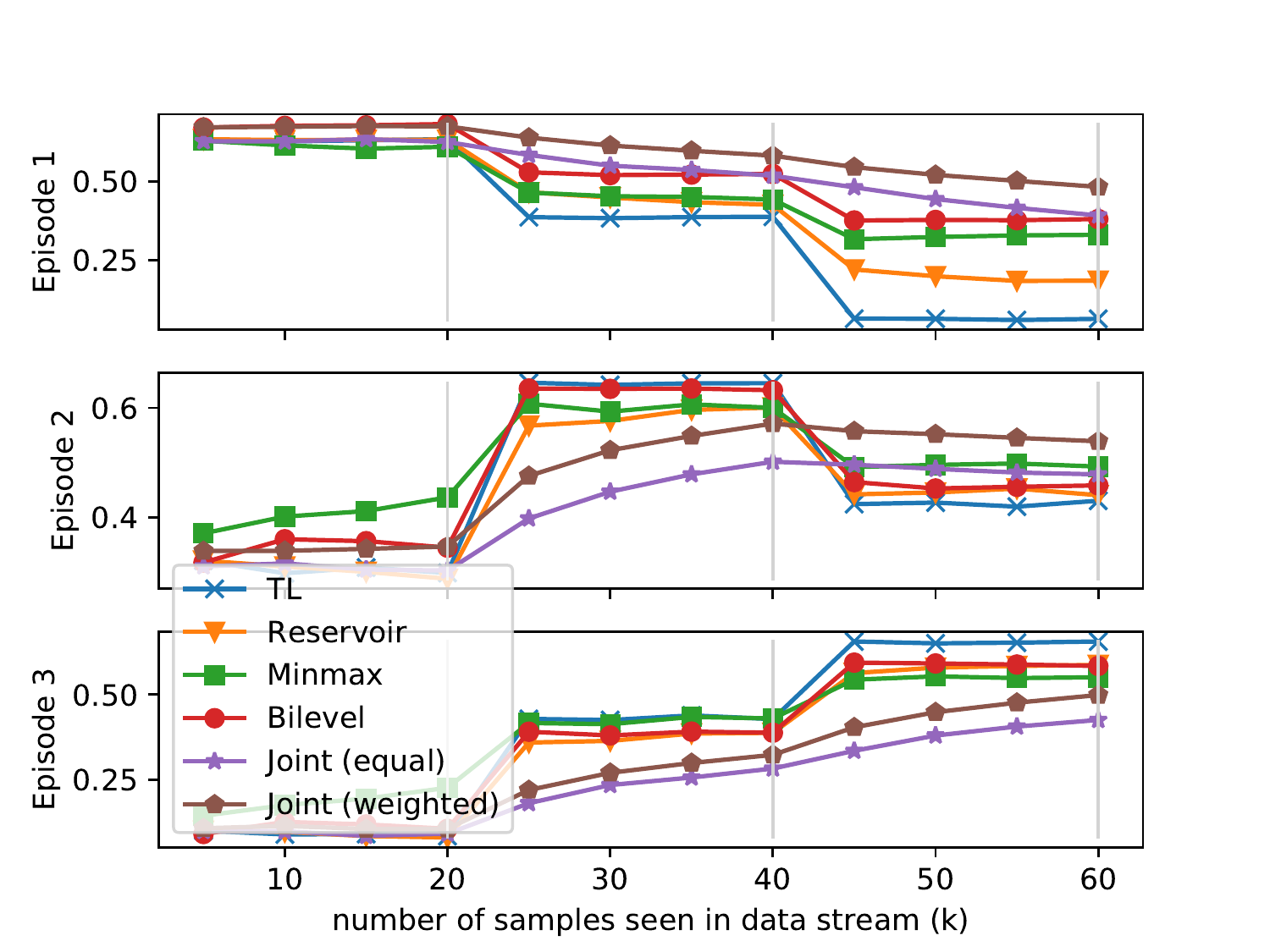}
    \footnotesize{(a) Test performance for each episode }
\end{minipage}
\noindent
\begin{minipage}[c]{0.4\linewidth}
    \centering
    \includegraphics[width=\textwidth]{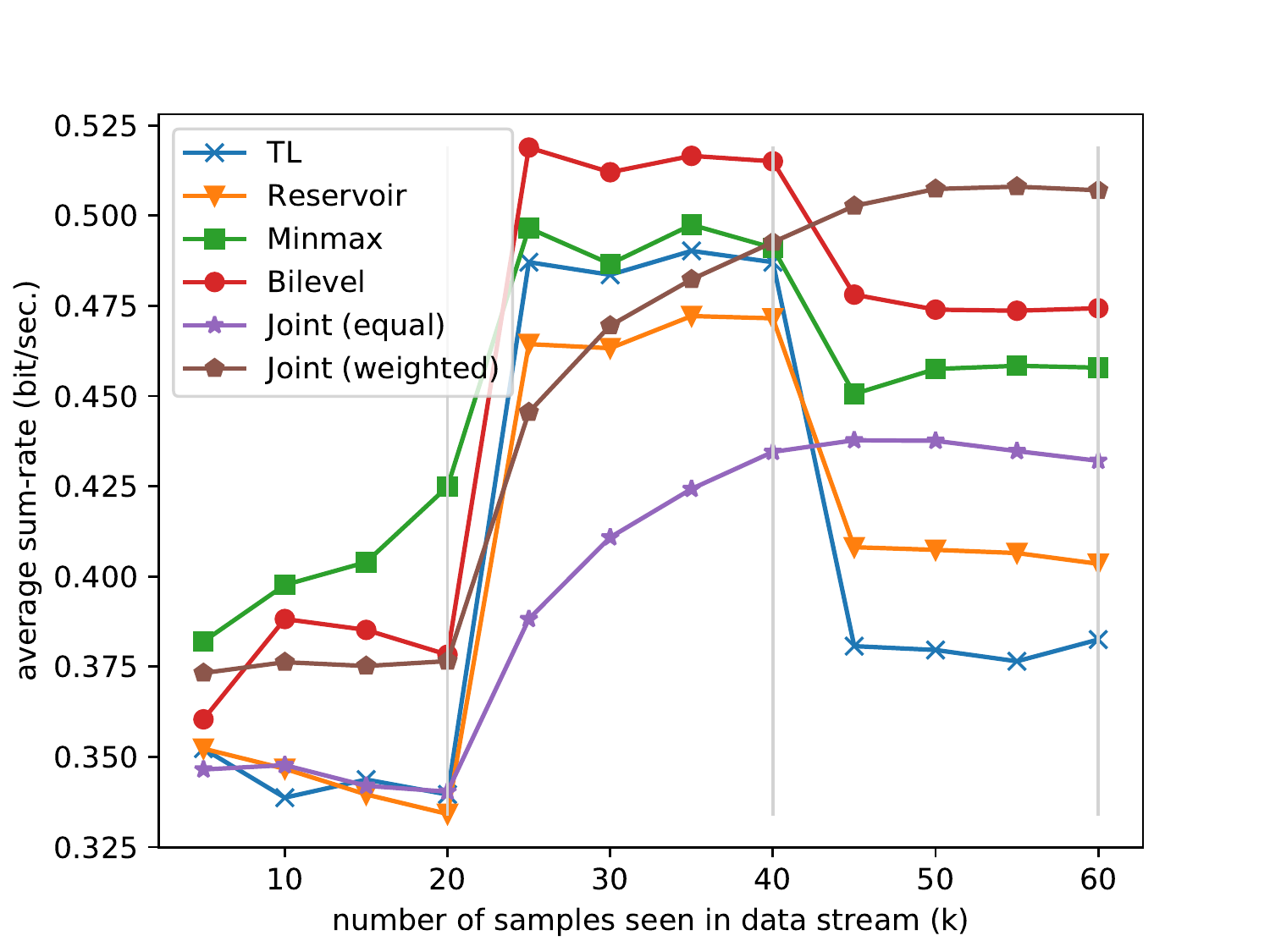}
    \footnotesize{(b) Average test performance for all episodes}
\end{minipage}
\caption{Average sum-rate comparison on real measured channels}
\label{mimo_fx3}
\end{figure}

We first show the system performance measured  by the achieved sum-rate for different approaches in Fig. \ref{mimo_fx3}. For each subplot of  Fig.  \ref{mimo_fx3} (a),  it displays a similar result as each subfigure in Fig. \ref{fig-random-rate} (a). 
It can be observed that, after experiencing all the samples $(x=60,000)$, our proposed fairness based method obtains reasonable sum-rate for all three episodes,  while the performance of both TL and reservoir sampling degrades when encountering test data from the old episodes. This can be attributed to the fact that the proposed method can focus on under-performing episodes (i.e. episode 1 and 2) while relaxing on outperforming episodes (i.e. episode 3). If we further average the sum-rate performance on all three episodes from Fig. \ref{mimo_fx3} (a), we obtained Fig. \ref{mimo_fx3} (b), in which it is clear that our proposed method is able to perform much better than TL and reservoir sampling.

Another interesting observation  (from subplot 3 of Fig. \ref{mimo_fx3} (a)) is that the proposed method is able to outperform the joint training (which uses the accumulated data) in terms of the average sum-rate, as can be seen in  Fig. \ref{mimo_fx3} (b) for $40,000 \le x \le 60,000$. One possible explanation is that the joint training will treat all samples equally, and thus only $1/3$ of training data will contribute to improve the performance of episode 3, resulting a slow adaption to new episodes. Instead, our proposed fairness based method  focuses more on data points that generate the highest cost, so it achieves higher average performance.

\begin{figure}
\centering
\begin{minipage}[c]{.4\linewidth}
    \centering
    \includegraphics[width=\textwidth]{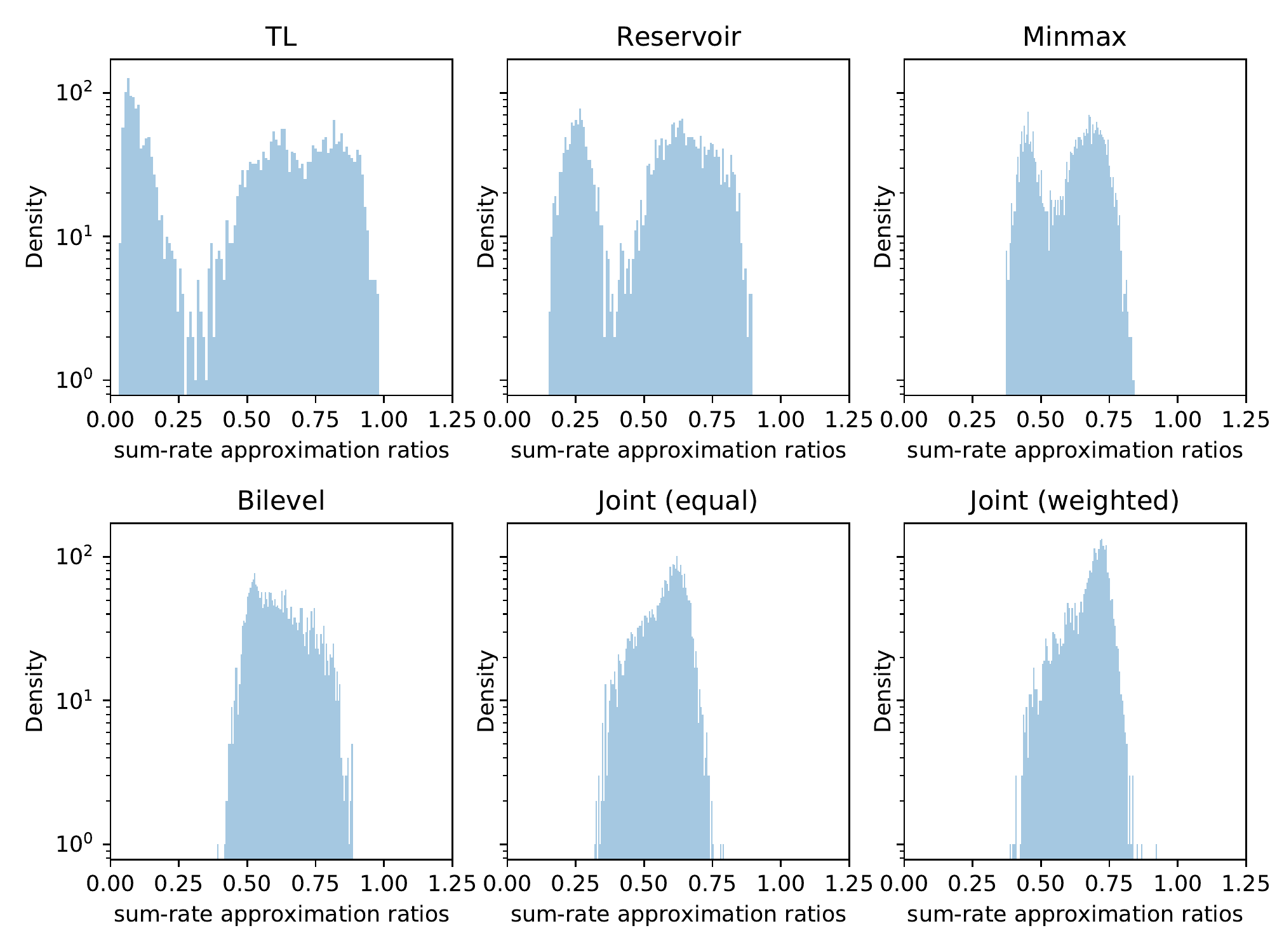}
    
    \footnotesize{(a) Probability Density Functions (PDF)}
\end{minipage}
\noindent
\begin{minipage}[c]{.4\linewidth}
    \centering
    \includegraphics[width=\textwidth]{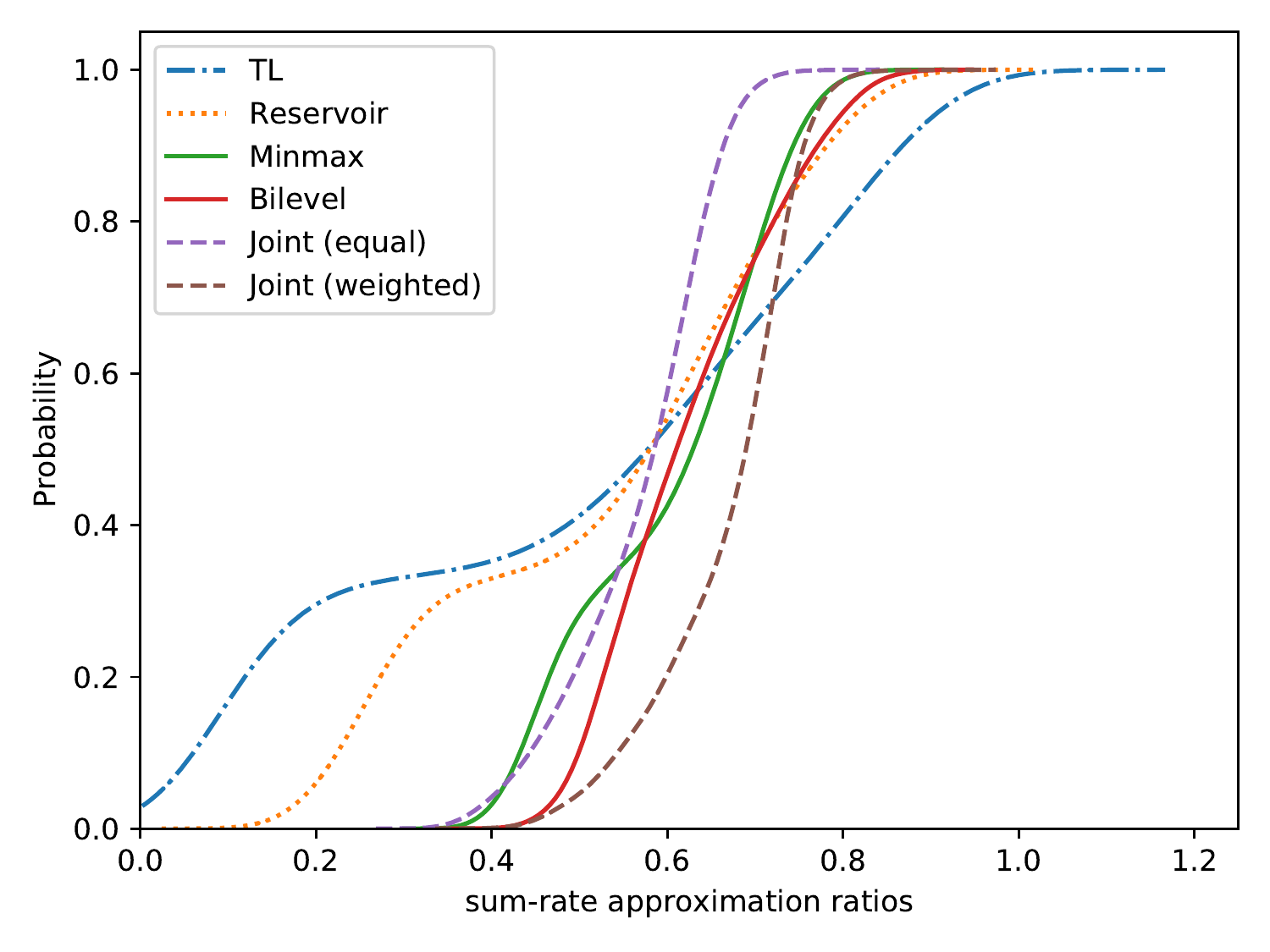}

    \footnotesize{(b) Cumulative Distribution Function (CDF)}
\end{minipage}
\caption{{\bf Fairness Comparison.} (a) PDF and (b) CDF distribution of the per-sample sum-rate  ratio evaluated at the last time stamp $(x=60,000)$ on the test set of all three episodes.  }
\label{mimo_fx3_fairness}
\end{figure}

\subsubsection{Fairness comparison} 
Next, we show that the proposed CL method outperforms other CL-based methods, not only in terms of the average sum-rate, but also in the sample fairness over all tasks. In Fig. \ref{mimo_fx3_fairness}, we show the test data  sum-rate ratio distributions for the final models generated by different approaches (i.e., when all the models have seen all $60,000$ data points).
 Specifically, the sum-rate ratio $R(\pi(\bTheta, \bx^{(i)}); \bx^{(i)}) / {\bar{R}(\bx^{(i)})}$ is computed according to Remark \ref{remark1}. That is, for a given test sample $\bx^{(i)}$, we divide the achievable sum-rate generated from the learning model, by what is achievable by the WMMSE algorithm \cite{shi2011iteratively}. It can be observed that our proposed approach contains fewer samples in the low sum-rate region, while TL and reservoir sampling perform worse on those data points. This result suggests that that the proposed approach indeed incorporates the problem structure and advocates fairness across the data samples.

\subsubsection{Gradual scenario change}
In the previous experiments, the sets of users from different episodes do not overlap with each other. That is, we were simulating scenarios where the environment is experiencing some {\it rapid} changes. In this subsection, we further simulate scenarios where the environment changes {\it slowly}. Towards this end,  we generate episodes such that the neighboring ones share some common areas. Specifically, we have five episodes, and users for episode 1 to 5 are drawn from columns C551-C1100, C826-C1375, C1101-C1650, C1376-C1925, and C1651-C2200, respectively. Simulation results are shown in Fig. \ref{fig:slow}. It can be observed that our proposed methods are still effective under this setting.

\begin{figure}
\centering
\begin{minipage}[c]{.4\linewidth}
    \centering
    \includegraphics[width=\textwidth]{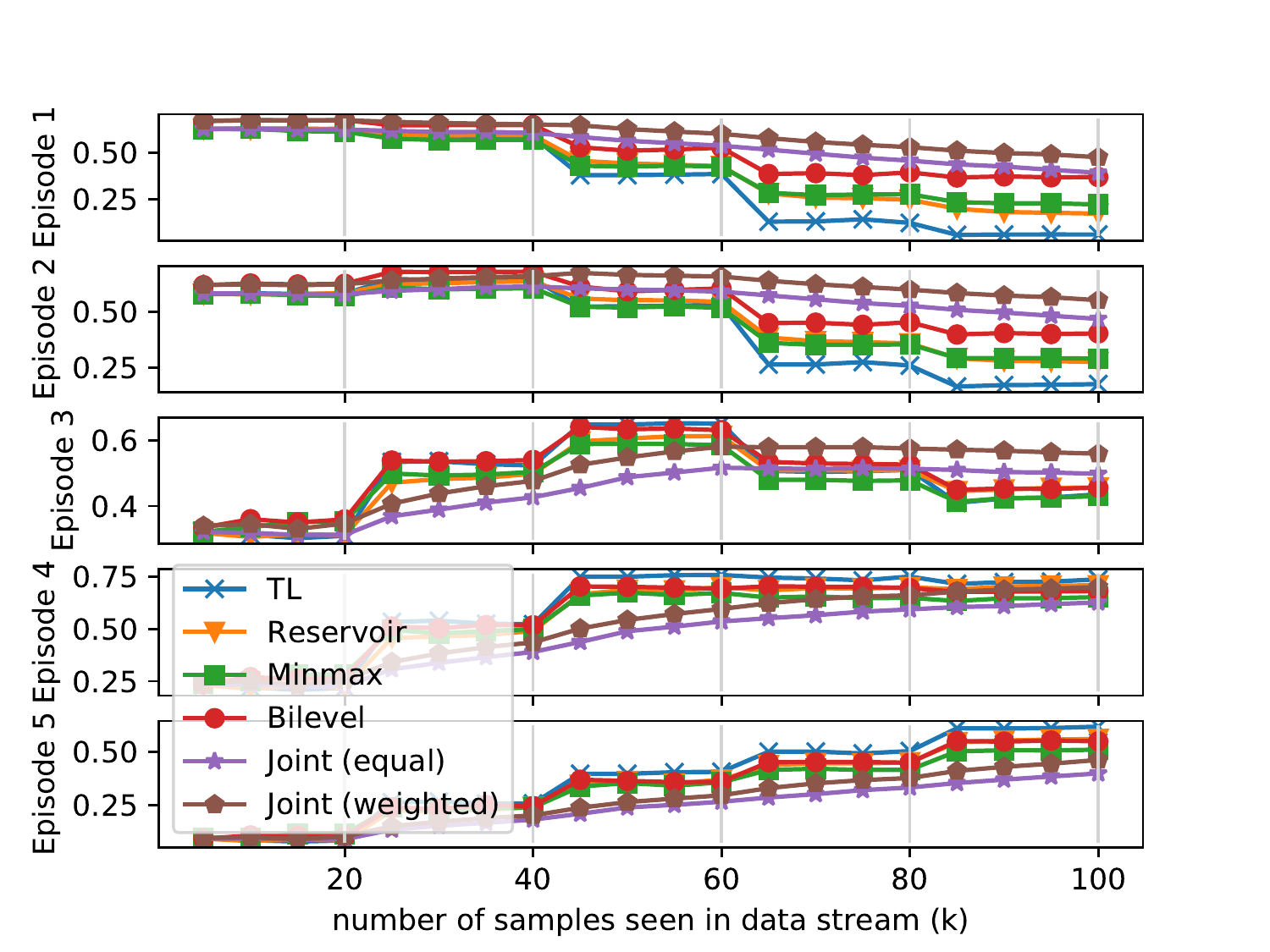}
    
    \footnotesize{(a) per-episode performance}
\end{minipage}
\begin{minipage}[c]{.4\linewidth}
    \centering
    \includegraphics[width=\textwidth]{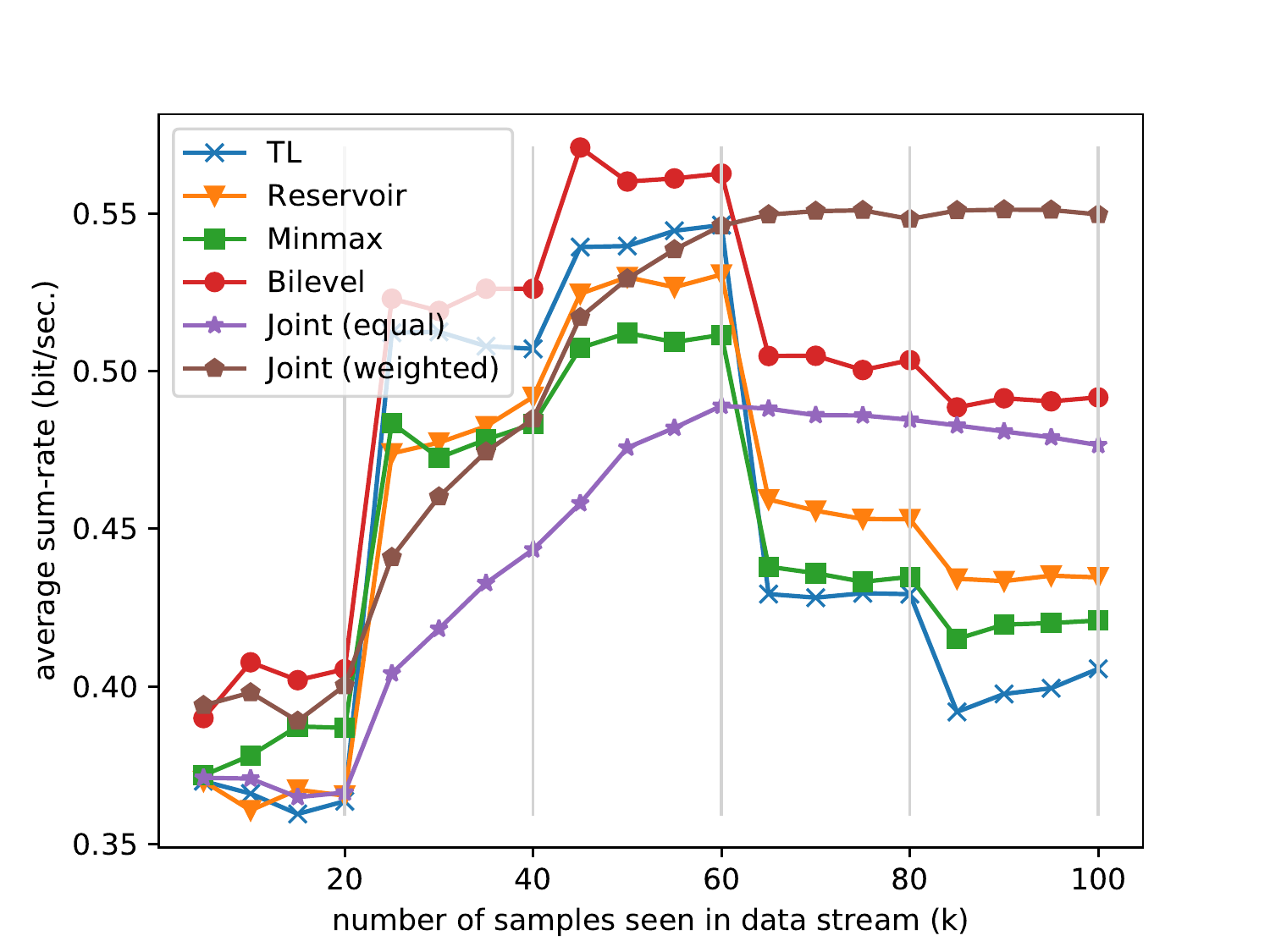}
    
    \footnotesize{(b) average performance}
\end{minipage}
\caption{Average sum-rate comparison on real measured channels over slowly changing environment}
\label{fig:slow}
\end{figure}

 \begin{figure*}
\centering
\begin{minipage}[c]{.3\linewidth}
    \centering
    \includegraphics[width=\textwidth]{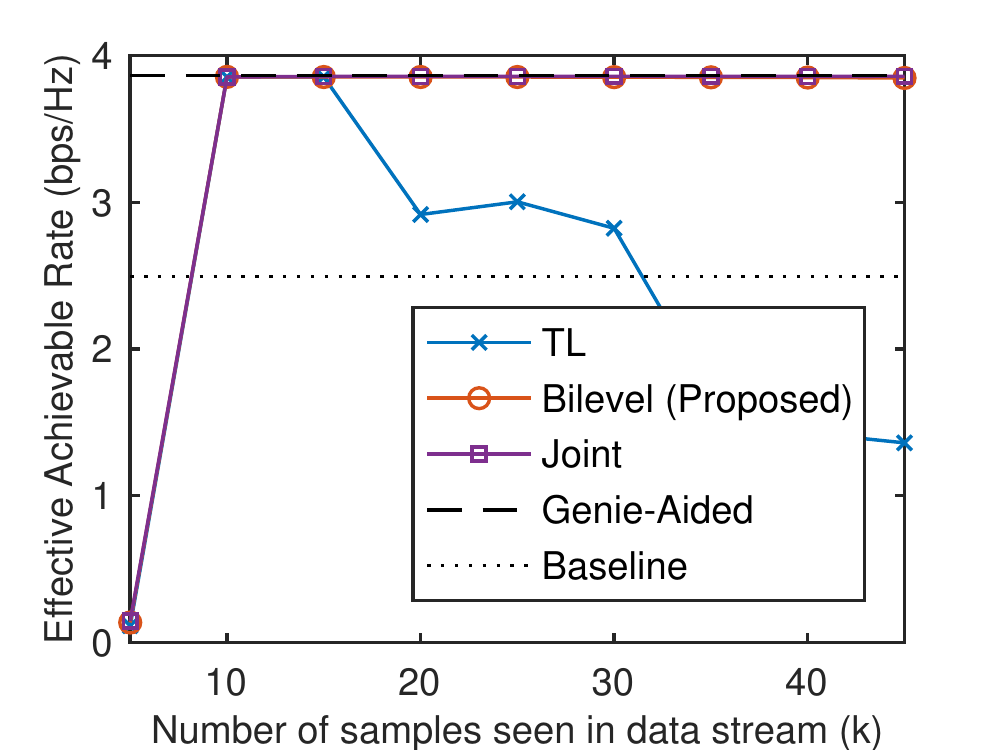}
    
    \footnotesize{(a) Test Results on Episode 1}
\end{minipage}
\begin{minipage}[c]{.3\linewidth}
    \centering
    \includegraphics[width=\textwidth]{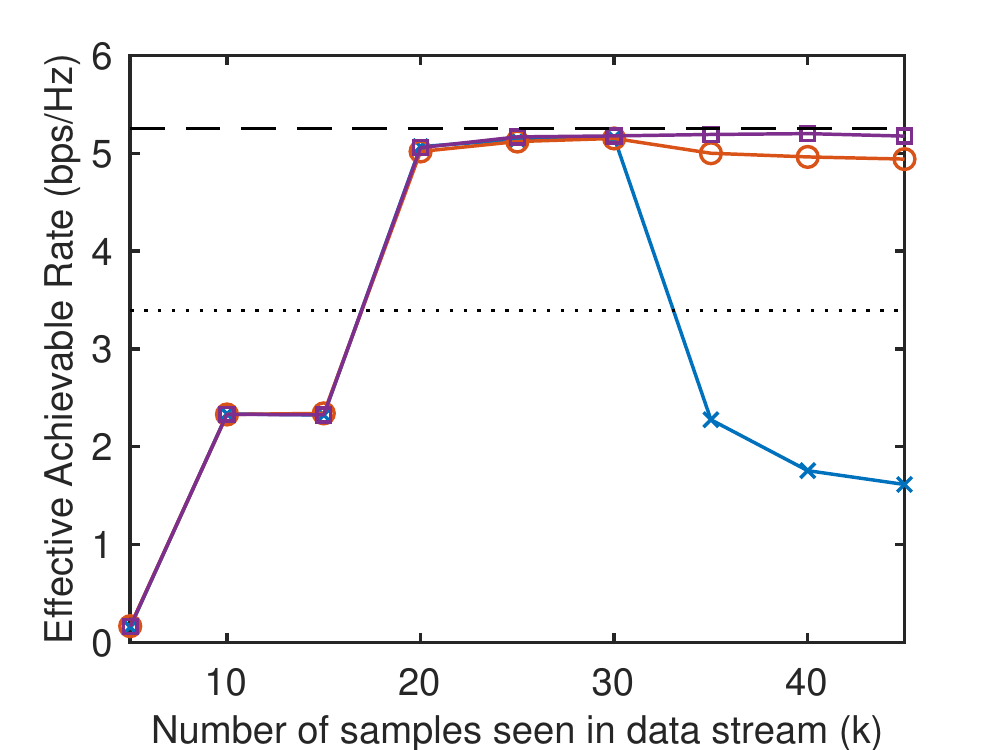}
    
    \footnotesize{(b) Test Results on Episode 2}
    
\end{minipage}
\begin{minipage}[c]{.3\linewidth}
    \centering
    \includegraphics[width=\textwidth]{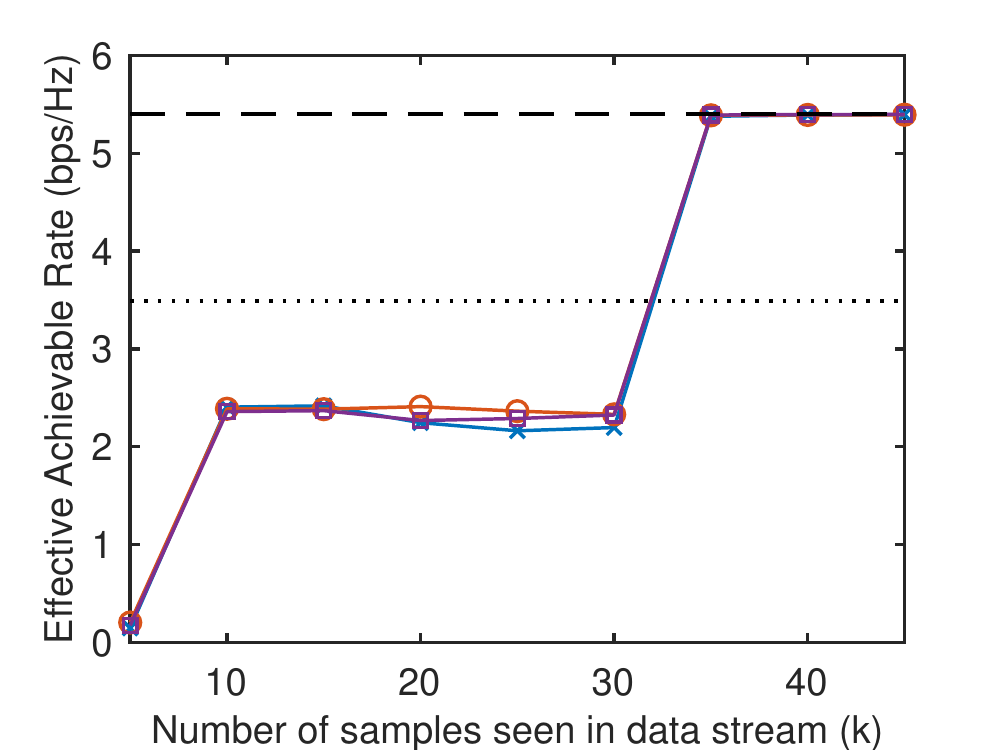}
    
    \footnotesize{(c) Test Results on Episode 3}
    
\end{minipage}

\caption{Achievable rate comparison of deep-learning coordinated beamforming strategies, genie-aided solution (perfectly knows the optimal beamforming vectors), and traditional mmWave beamforming techniques \cite{alkhateeb2018deep}.}
\label{beamform}
\end{figure*}

 \begin{figure}
\centering
\begin{minipage}[c]{.5\linewidth}
    \centering
    \includegraphics[width=\textwidth]{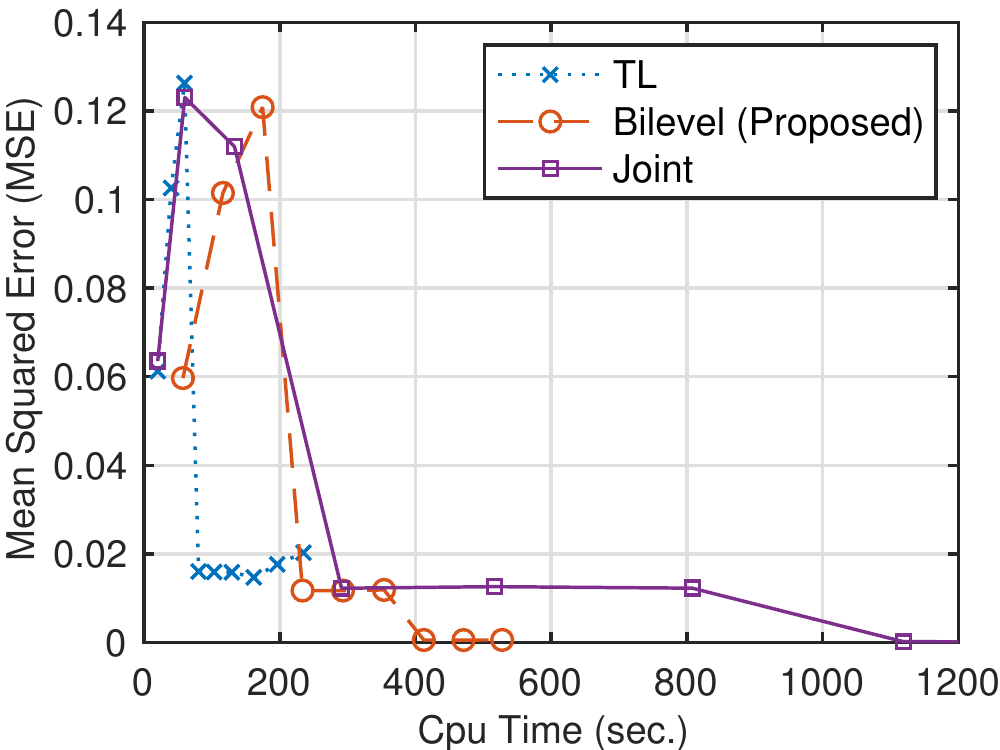}
\end{minipage}
\caption{ Training MSE performance over the mixed dataset of all episodes versus CPU time. }
\label{time}
\end{figure}

\subsection{Beamforming Experiments} \label{sec:beamform}
Next, we further validate our CL based approach, by applying it to a coordinated beamforming problem for the millimeter wave system, where a number of BSs are simultaneously serving one mobile user over the 60 GHz band \cite{alkhateeb2018deep}. Different from the previous sections where only single antenna is adopted, we consider the multi-antenna setup with four BSs (3,4,5,6 in Fig. \ref{location}), and each BS uses uniform planar array (UPA) consisting of a total of 256 antenna elements (32 columns and 8 rows), and use 30dBm transmit power.   

We adopt the problem formulation developed in \cite{sun2018learning,alkhateeb2018deep}, where the idea is to use the uplink pilot signal received at the terminal BSs with only omni or quasi-omni beam patterns to learn and predict the best RF beamforming vectors. 
The learning based method we adopt is the fully connected network as suggested in  \cite{sun2018learning,alkhateeb2018deep}.
By leveraging the intuition that the received signal renders an RF defining signature for the user location and its interaction with the surrounding environment, the authors of \cite{alkhateeb2018deep} showed that the DL solution performs almost as well as the genie-aided solution that perfectly knows the optimal beamforming vectors.

In our simulation, we define three episodes, where the user distributions are drawn from columns  C551 - C650, C826 - C925, C1101 - C1200, respectively. For each episode, we generate $15,000$ samples for training and $1,000$ for testing. The simulation results over different approaches are compared and reported in Fig. \ref{beamform}, where the $x$-axis represents the number of data samples that has been observed, and the $y$-axis denotes the  effective achievable rate. For our proposed approach (Bilevel), both the training loss $\ell(\cdot)$ and the system performance loss $u(\cdot)$ are chosen as the MSE loss \eqref{MSE-loss}. It can be observed that the proposed algorithm (Bilevel) almost matches the joint training and the optimal genie-aided performances, with only limited memory usage, and it outperforms the TL approach.  

 Lastly, we compare the computational cost of all methods during the entire training stage. We record the required training time for all approaches when they experiencing all three episodes, then plot their achieved training loss (i.e., the MSE loss in this case,  evaluated on the mixture dataset of all three episodes) versus the consumed cpu time in Fig. \ref{time}. It can be observed that the joint training and proposed CL approach can achieve zero training loss for all episodes (after 1,100 and 400 seconds, respectively), while the TL approach can never achieve zero training loss for all episodes although it takes less time. The proposed fairness based CL methods strike a good balance between the time complexity and the prediction accuracy.

\section{Conclusion and Future Works} \label{conclusion}

In this work, we design a new ``learning to continuously optimize'' framework for optimizing wireless resources in dynamic environments, where parameters such as CSIs keep changing. By introducing  continual learning (CL) into the modeling process, our framework is able to seamlessly  and  efficiently  adapt  to  the  episodically dynamic  environment,  without  knowing  the  episode boundary, and most importantly, maintain high performance over all the previously encountered scenarios. 
The proposed approach is validated through two popular wireless resource allocation problems (one for power control and one for beamforming), 
and uses both synthetic and ray-tracing based data sets. 
 Simulation results show that our framework is consistently better than naive transfer learning method, and it achieves better performance than classical CL based approaches.
Our empirical results make us believe  that the proposed approaches can be extended to many other related problems.

Our work represents a first step towards understanding the capability of deep learning for wireless problems with dynamic environments. There are many interesting questions to be addressed in the future, such as:
\begin{itemize}
    \item Is it possible to design theoretical results for the proposed stage-wise bilevel optimization problem \eqref{bilevel-ep} or even the global bilevel optimization problem \eqref{eq:bilevel}?
    \item Is it possible to quantify the generalization performance of the proposed fairness framework?
    \item Is it possible to extend our frameworks to other wireless tasks such as signal detection, channel estimation and CSI compression?
\end{itemize}

\bibliographystyle{IEEEtran}
\bibliography{CL,WC,ref}

\begin{thebibliography}{10}
\providecommand{\url}[1]{#1}
\csname url@samestyle\endcsname
\providecommand{\newblock}{\relax}
\providecommand{\bibinfo}[2]{#2}
\providecommand{\BIBentrySTDinterwordspacing}{\spaceskip=0pt\relax}
\providecommand{\BIBentryALTinterwordstretchfactor}{4}
\providecommand{\BIBentryALTinterwordspacing}{\spaceskip=\fontdimen2\font plus
\BIBentryALTinterwordstretchfactor\fontdimen3\font minus
  \fontdimen4\font\relax}
\providecommand{\BIBforeignlanguage}[2]{{%
\expandafter\ifx\csname l@#1\endcsname\relax
\typeout{** WARNING: IEEEtran.bst: No hyphenation pattern has been}%
\typeout{** loaded for the language `#1'. Using the pattern for}%
\typeout{** the default language instead.}%
\else
\language=\csname l@#1\endcsname
\fi
#2}}
\providecommand{\BIBdecl}{\relax}
\BIBdecl

\bibitem{sun2021}
H.~Sun, W.~Pu, M.~Zhu, X.~Fu, T.-H. Chang, and M.~Hong, ``Learning to
  continuously optimize wireless resource in episodically dynamic
  environment,'' in \emph{proceedings of the International Conference on
  Acoustics, Speech, and Signal Processing (ICASSP)}, 2021.

\bibitem{voulodimos2018deep}
A.~Voulodimos, N.~Doulamis, A.~Doulamis, and E.~Protopapadakis, ``Deep learning
  for computer vision: A brief review,'' \emph{Computational Intelligence and
  Neuroscience}, 2018.

\bibitem{young2018recent}
T.~Young, D.~Hazarika, S.~Poria, and E.~Cambria, ``Recent trends in deep
  learning based natural language processing,'' \emph{IEEE Computational
  Intelligence Magazine}, vol.~13, no.~3, pp. 55--75, 2018.

\bibitem{wang2015collaborative}
H.~Wang, N.~Wang, and D.-Y. Yeung, ``Collaborative deep learning for
  recommender systems,'' in \emph{Proceedings of the 21th ACM International
  Conference on Knowledge Discovery and Data Mining (SIGKDD)}, 2015, pp.
  1235--1244.

\bibitem{goodfellow2016deep}
I.~Goodfellow, Y.~Bengio, A.~Courville, and Y.~Bengio, \emph{Deep
  learning}.\hskip 1em plus 0.5em minus 0.4em\relax MIT press Cambridge, 2016,
  vol.~1, no.~2.

\bibitem{ye2017power}
H.~Ye, G.~Y. Li, and B.-H. Juang, ``Power of deep learning for channel
  estimation and signal detection in {OFDM} systems,'' \emph{IEEE Wireless
  Communications Letters}, vol.~7, no.~1, pp. 114--117, 2017.

\bibitem{sun2019deep}
H.~Sun, A.~O. Kaya, M.~Macdonald, H.~Viswanathan, and M.~Hong, ``Deep learning
  based preamble detection and {TOA} estimation,'' in \emph{proceedings of the
  2019 IEEE Global Communications Conference (GLOBECOM)}.

\bibitem{nachmani2018deep}
E.~Nachmani, E.~Marciano, L.~Lugosch, W.~J. Gross, D.~Burshtein, and
  Y.~Be’ery, ``Deep learning methods for improved decoding of linear codes,''
  \emph{IEEE Journal of Selected Topics in Signal Processing}, vol.~12, no.~1,
  pp. 119--131, 2018.

\bibitem{wen2018deep}
C.-K. Wen, W.-T. Shih, and S.~Jin, ``Deep learning for massive {MIMO} {CSI}
  feedback,'' \emph{IEEE Wireless Communications Letters}, vol.~7, no.~5, pp.
  748--751, 2018.

\bibitem{sun2018limited}
H.~Sun, Z.~Zhao, X.~Fu, and M.~Hong, ``Limited feedback double directional
  massive {MIMO} channel estimation: From low-rank modeling to deep learning,''
  in \emph{proceedings of the IEEE 19th International Workshop on Signal
  Processing Advances in Wireless Communications (SPAWC)}, 2018.

\bibitem{dorner2017deep}
S.~D{\"o}rner, S.~Cammerer, J.~Hoydis, and S.~Ten~Brink, ``Deep learning based
  communication over the air,'' \emph{IEEE Journal of Selected Topics in Signal
  Processing}, vol.~12, no.~1, pp. 132--143, 2017.

\bibitem{o2017introduction}
T.~O’Shea and J.~Hoydis, ``An introduction to deep learning for the physical
  layer,'' \emph{IEEE Transactions on Cognitive Communications and Networking},
  vol.~3, no.~4, pp. 563--575, 2017.

\bibitem{sun2018learning}
H.~Sun, X.~Chen, Q.~Shi, M.~Hong, X.~Fu, and N.~D. Sidiropoulos, ``Learning to
  optimize: Training deep neural networks for interference management,''
  \emph{IEEE Transactions on Signal Processing}, vol.~66, no.~20, pp.
  5438--5453, 2018.

\bibitem{lee2018deep}
W.~Lee, M.~Kim, and D.-H. Cho, ``Deep power control: Transmit power control
  scheme based on convolutional neural network,'' \emph{IEEE Communications
  Letters}, vol.~22, no.~6, pp. 1276--1279, 2018.

\bibitem{liang2019towards}
F.~Liang, C.~Shen, W.~Yu, and F.~Wu, ``Towards optimal power control via
  ensembling deep neural networks,'' \emph{IEEE Transactions on
  Communications}, vol.~68, no.~3, pp. 1760--1776, 2019.

\bibitem{eisen2020optimal}
M.~Eisen and A.~R. Ribeiro, ``Optimal wireless resource allocation with random
  edge graph neural networks,'' \emph{IEEE Transactions on Signal Processing},
  2020.

\bibitem{shen2019lorm}
Y.~Shen, Y.~Shi, J.~Zhang, and K.~B. Letaief, ``{LORM}: Learning to optimize
  for resource management in wireless networks with few training samples,''
  \emph{IEEE Transactions on Wireless Communications}, vol.~19, no.~1, pp.
  665--679, 2019.

\bibitem{huang2019fast}
H.~Huang, Y.~Peng, J.~Yang, W.~Xia, and G.~Gui, ``Fast beamforming design via
  deep learning,'' \emph{IEEE Transactions on Vehicular Technology}, vol.~69,
  no.~1, pp. 1065--1069, 2019.

\bibitem{cui2019spatial}
W.~Cui, K.~Shen, and W.~Yu, ``Spatial deep learning for wireless scheduling,''
  \emph{IEEE Journal on Selected Areas in Communications}, vol.~37, no.~6, pp.
  1248--1261, 2019.

\bibitem{parisi2019continual}
G.~I. Parisi, R.~Kemker, J.~L. Part, C.~Kanan, and S.~Wermter, ``Continual
  lifelong learning with neural networks: A review,'' \emph{Neural Networks},
  vol. 113, pp. 54--71, 2019.

\bibitem{mccloskey1989catastrophic}
M.~McCloskey and N.~J. Cohen, ``Catastrophic interference in connectionist
  networks: The sequential learning problem,'' in \emph{Psychology of Learning
  and Motivation}.\hskip 1em plus 0.5em minus 0.4em\relax Elsevier, 1989,
  vol.~24, pp. 109--165.

\bibitem{kirkpatrick2017overcoming}
J.~Kirkpatrick, R.~Pascanu, N.~Rabinowitz, J.~Veness, G.~Desjardins, A.~A.
  Rusu, K.~Milan, J.~Quan, T.~Ramalho, A.~Grabska-Barwinska \emph{et~al.},
  ``Overcoming catastrophic forgetting in neural networks,'' \emph{Proceedings
  of the National Academy of Sciences}, vol. 114, no.~13, pp. 3521--3526, 2017.

\bibitem{li2017learning}
Z.~Li and D.~Hoiem, ``Learning without forgetting,'' \emph{IEEE Transactions on
  Pattern Analysis and Machine Intelligence}, vol.~40, no.~12, pp. 2935--2947,
  2017.

\bibitem{yoon2018lifelong}
J.~Yoon, E.~Yang, J.~Lee, and S.~J. Hwang, ``Lifelong learning with dynamically
  expandable networks,'' in \emph{International Conference on Learning
  Representations}, 2018.

\bibitem{zenke2017continual}
F.~Zenke, B.~Poole, and S.~Ganguli, ``Continual learning through synaptic
  intelligence,'' \emph{Proceedings of Machine Learning Research}, vol.~70, p.
  3987, 2017.

\bibitem{rusu2016progressive}
A.~A. Rusu, N.~C. Rabinowitz, G.~Desjardins, H.~Soyer, J.~Kirkpatrick,
  K.~Kavukcuoglu, R.~Pascanu, and R.~Hadsell, ``Progressive neural networks,''
  \emph{arXiv:1606.04671}, 2016.

\bibitem{lopez2017gradient}
D.~Lopez-Paz and M.~Ranzato, ``Gradient episodic memory for continual
  learning,'' in \emph{Advances in Neural Information Processing Systems},
  2017, pp. 6467--6476.

\bibitem{shin2017continual}
H.~Shin, J.~K. Lee, J.~Kim, and J.~Kim, ``Continual learning with deep
  generative replay,'' in \emph{Advances in Neural Information Processing
  Systems}, 2017, pp. 2990--2999.

\bibitem{rebuffi2017icarl}
S.-A. Rebuffi, A.~Kolesnikov, G.~Sperl, and C.~H. Lampert, ``icarl: Incremental
  classifier and representation learning,'' in \emph{Proceedings of the IEEE
  conference on Computer Vision and Pattern Recognition}, 2017, pp. 2001--2010.

\bibitem{isele2018selective}
D.~Isele and A.~Cosgun, ``Selective experience replay for lifelong learning,''
  in \emph{Proceedings of the AAAI Conference on Artificial Intelligence
  (AAAI)}, 2018, pp. 3302--3309.

\bibitem{aljundi2019gradient}
R.~Aljundi, M.~Lin, B.~Goujaud, and Y.~Bengio, ``Gradient based sample
  selection for online continual learning,'' in \emph{Advances in Neural
  Information Processing Systems}, 2019, pp. 11\,816--11\,825.

\bibitem{rolnick2019experience}
D.~Rolnick, A.~Ahuja, J.~Schwarz, T.~Lillicrap, and G.~Wayne, ``Experience
  replay for continual learning,'' in \emph{Advances in Neural Information
  Processing Systems}, 2019, pp. 350--360.

\bibitem{shi2011iteratively}
Q.~Shi, M.~Razaviyayn, Z.-Q. Luo, and C.~He, ``An iteratively weighted {MMSE}
  approach to distributed sum-utility maximization for a {MIMO} interfering
  broadcast channel,'' \emph{IEEE Transactions on Signal Processing}, vol.~59,
  no.~9, pp. 4331--4340, 2011.

\bibitem{shen2019graph}
Y.~Shen, Y.~Shi, J.~Zhang, and K.~B. Letaief, ``A graph neural network approach
  for scalable wireless power control,'' in \emph{proceedings of the 2019 IEEE
  Globecom Workshops (GC Wkshps)}.

\bibitem{nasir2019multi}
Y.~S. Nasir and D.~Guo, ``Multi-agent deep reinforcement learning for dynamic
  power allocation in wireless networks,'' \emph{IEEE Journal on Selected Areas
  in Communications}, vol.~37, no.~10, pp. 2239--2250, 2019.

\bibitem{balatsoukas2019deep}
A.~Balatsoukas-Stimming and C.~Studer, ``Deep unfolding for communications
  systems: A survey and some new directions,'' in \emph{proceedings of the IEEE
  International Workshop on Signal Processing Systems (SiPS)}, 2019, pp.
  266--271.

\bibitem{gregor2010learning}
K.~Gregor and Y.~LeCun, ``Learning fast approximations of sparse coding,'' in
  \emph{proceedings of the 27th International Conference on Machine Learning
  (ICML)}, 2010, pp. 399--406.

\bibitem{hershey2014deep}
J.~R. Hershey, J.~L. Roux, and F.~Weninger, ``Deep unfolding: Model-based
  inspiration of novel deep architectures,'' \emph{arXiv preprint
  arXiv:1409.2574}, 2014.

\bibitem{sprechmann2013supervised}
P.~Sprechmann, R.~Litman, T.~B. Yakar, A.~M. Bronstein, and G.~Sapiro,
  ``Supervised sparse analysis and synthesis operators,'' in \emph{Advances in
  Neural Information Processing Systems}, 2013, pp. 908--916.

\bibitem{samuel2017deep}
N.~Samuel, T.~Diskin, and A.~Wiesel, ``Deep mimo detection,'' \emph{arXiv
  preprint arXiv:1706.01151}, 2017.

\bibitem{samuel2019learning}
------, ``Learning to detect,'' \emph{IEEE Transactions on Signal Processing},
  vol.~67, no.~10, pp. 2554--2564, 2019.

\bibitem{he2018model}
H.~He, C.-K. Wen, S.~Jin, and G.~Y. Li, ``A model-driven deep learning network
  for mimo detection,'' in \emph{proceedings of the IEEE Global Conference on
  Signal and Information Processing (GlobalSIP)}, 2018, pp. 584--588.

\bibitem{cammerer2017scaling}
S.~Cammerer, T.~Gruber, J.~Hoydis, and S.~Ten~Brink, ``Scaling deep
  learning-based decoding of polar codes via partitioning,'' in
  \emph{proceedings of the IEEE Global Communications Conference (GLOBECOM)},
  2017, pp. 1--6.

\bibitem{eisen2019learning}
M.~Eisen, C.~Zhang, L.~F. Chamon, D.~D. Lee, and A.~Ribeiro, ``Learning optimal
  resource allocations in wireless systems,'' \emph{IEEE Transactions on Signal
  Processing}, vol.~67, no.~10, pp. 2775--2790, 2019.

\bibitem{borgerding2017amp}
M.~Borgerding, P.~Schniter, and S.~Rangan, ``Amp-inspired deep networks for
  sparse linear inverse problems,'' \emph{IEEE Transactions on Signal
  Processing}, vol.~65, no.~16, pp. 4293--4308, 2017.

\bibitem{hu2020iterative}
Q.~Hu, Y.~Cai, Q.~Shi, K.~Xu, G.~Yu, and Z.~Ding, ``Iterative algorithm induced
  deep-unfolding neural networks: Precoding design for multiuser mimo
  systems,'' \emph{IEEE Transactions on Wireless Communications}, 2020.

\bibitem{ring1994continual}
M.~B. Ring, ``Continual learning in reinforcement environments,'' Ph.D.
  dissertation, University of Texas at Austin Austin, Texas 78712, 1994.

\bibitem{robins1995catastrophic}
A.~Robins, ``Catastrophic forgetting, rehearsal and pseudorehearsal,''
  \emph{Connection Science}, vol.~7, no.~2, pp. 123--146, 1995.

\bibitem{shalev2011online}
S.~Shalev-Shwartz \emph{et~al.}, ``Online learning and online convex
  optimization,'' 2011.

\bibitem{aljundi2019online}
R.~Aljundi, E.~Belilovsky, T.~Tuytelaars, L.~Charlin, M.~Caccia, M.~Lin, and
  L.~Page-Caccia, ``Online continual learning with maximal interfered
  retrieval,'' in \emph{Advances in Neural Information Processing Systems},
  2019, pp. 11\,849--11\,860.

\bibitem{pan2009survey}
S.~J. Pan and Q.~Yang, ``A survey on transfer learning,'' \emph{IEEE
  Transactions on Knowledge and Data Engineering}, vol.~22, no.~10, pp.
  1345--1359, 2009.

\bibitem{yuan2020transfer}
Y.~Yuan, G.~Zheng, K.-K. Wong, B.~Ottersten, and Z.-Q. Luo, ``Transfer learning
  and meta learning based fast downlink beamforming adaptation,'' \emph{arXiv
  preprint arXiv:2011.00903}, 2020.

\bibitem{zappone2019model}
A.~Zappone, M.~Di~Renzo, M.~Debbah, T.~T. Lam, and X.~Qian, ``Model-aided
  wireless artificial intelligence: Embedding expert knowledge in deep neural
  networks for wireless system optimization,'' \emph{IEEE Vehicular Technology
  Magazine}, vol.~14, no.~3, pp. 60--69, 2019.

\bibitem{luo2008dynamic}
Z.-Q. Luo and S.~Zhang, ``Dynamic spectrum management: Complexity and
  duality,'' \emph{IEEE Journal of Selected Topics in Signal Processing},
  vol.~2, no.~1, pp. 57--73, 2008.

\bibitem{hayes2019memory}
T.~L. Hayes, N.~D. Cahill, and C.~Kanan, ``Memory efficient experience replay
  for streaming learning,'' in \emph{2019 International Conference on Robotics
  and Automation (ICRA)}.\hskip 1em plus 0.5em minus 0.4em\relax IEEE, 2019,
  pp. 9769--9776.

\bibitem{mo2000}
J.~Mo and J.~Walrand, ``Fair end-to-end window-based congestion control,''
  \emph{IEEE/ACM Transactions on Networking,}, vol.~8, no.~5, pp. 556 --567,
  2000.

\bibitem{hong12survey}
M.~Hong and Z.-Q. Luo, ``Signal processing and optimal resource allocation for
  the interference channel,'' in \emph{Academic Press Library in Signal
  Processing}.\hskip 1em plus 0.5em minus 0.4em\relax Academic Press, 2013.

\bibitem{song2021}
B.~Song, H.~Sun, W.~Pu, S.~Liu, and M.~Hong, ``To supervise or not to
  supervise: How to effectively learn wireless interference management
  models?'' in \emph{proceedings of the International Conference on Acoustics,
  Speech, and Signal Processing (ICASSP)}, 2021.

\bibitem{Bengtsson99optimaldownlink}
M.~Bengtsson and B.~Ottersten, ``Optimal downlink beamforming using
  semidefinite optimization,'' in \emph{Proceedings of the 37th Annual Allerton
  Conference}, 1999.

\bibitem{Razaviyayn12maxmin}
M.~Razaviyayn, M.~Hong, and Z.-Q. Luo, ``Linear transceiver design for a {MIMO}
  interfering broadcast channel achieving max-min fairness,'' \emph{Signal
  Processing}, vol.~93, no.~12, pp. 3327--3340, 2013.

\bibitem{lin2020gradient}
T.~Lin, C.~Jin, and M.~Jordan, ``On gradient descent ascent for
  nonconvex-concave minimax problems,'' in \emph{International Conference on
  Machine Learning}.\hskip 1em plus 0.5em minus 0.4em\relax PMLR, 2020, pp.
  6083--6093.

\bibitem{razaviyayn2020nonconvex}
M.~Razaviyayn, T.~Huang, S.~Lu, M.~Nouiehed, M.~Sanjabi, and M.~Hong,
  ``Nonconvex min-max optimization: Applications, challenges, and recent
  theoretical advances,'' \emph{IEEE Signal Processing Magazine}, vol.~37,
  no.~5, pp. 55--66, 2020.

\bibitem{hong2020two}
M.~Hong, H.-T. Wai, Z.~Wang, and Z.~Yang, ``A two-timescale framework for
  bilevel optimization: Complexity analysis and application to actor-critic,''
  \emph{arXiv preprint arXiv:2007.05170}, 2020.

\bibitem{herrera2020estimating}
C.~Herrera, F.~Krach, and J.~Teichmann, ``Estimating full lipschitz constants
  of deep neural networks,'' \emph{arXiv preprint arXiv:2004.13135}, 2020.

\bibitem{nesterov2003introductory}
Y.~Nesterov, \emph{Introductory lectures on convex optimization: A basic
  course}.\hskip 1em plus 0.5em minus 0.4em\relax Springer Science \& Business
  Media, 2003, vol.~87.

\bibitem{wang2017stochastic}
M.~Wang, E.~X. Fang, and H.~Liu, ``Stochastic compositional gradient descent:
  algorithms for minimizing compositions of expected-value functions,''
  \emph{Mathematical Programming}, vol. 161, no. 1-2, pp. 419--449, 2017.

\bibitem{chen2020solving}
T.~Chen, Y.~Sun, and W.~Yin, ``Solving stochastic compositional optimization is
  nearly as easy as solving stochastic optimization,'' \emph{arXiv preprint
  arXiv:2008.10847}, 2020.

\bibitem{alkhateeb2018deep}
A.~Alkhateeb, S.~Alex, P.~Varkey, Y.~Li, Q.~Qu, and D.~Tujkovic, ``Deep
  learning coordinated beamforming for highly-mobile millimeter wave systems,''
  \emph{IEEE Access}, vol.~6, pp. 37\,328--37\,348, 2018.

\bibitem{Alkhateeb2019}
A.~Alkhateeb, ``{DeepMIMO}: A generic deep learning dataset for millimeter wave
  and massive {MIMO} applications,'' in \emph{proceedings of the Information
  Theory and Applications Workshop (ITA)}, San Diego, CA, 2019.

\end{thebibliography}

\newpage
\section{Proof of Theorem \ref{theorem1}}
\begin{proof}\label{eq.theorem1-proof}
First, we need to establish the gradient smoothness condition of the  compositional function $\bF(\bTheta)$ as defined in \eqref{eq.stochastic.formulation}. That is, for some $L>0$, the following holds:
$$\|\nabla \bF(\bTheta) - \nabla \bF(\bTheta') \| \le L \|\bTheta - \bTheta'\|,$$
where the gradient is computed as
\begin{align*} 
    \nabla \bF(\bTheta)= &\nabla \dg(\bTheta) \nabla_1 \df(\dg(\bTheta), \bTheta)  + \nabla_2 \df(\dg(\bTheta), \bTheta).
\end{align*}
 The proof is relegated to Lemma \ref{lemma-lip-stoc} in the supplemental material \ref{sec-lemma} for completeness. 

Then, using the smoothness of $\nabla \bF(\bTheta^k)$, we have
{\begin{align*}
    & \ \ \ \bF(\bTheta^{k+1})\\
    &\leq \bF(\bTheta^k)+\dotp{\nabla \bF(\bTheta^k),\bTheta^{k+1}-\bTheta^k}+\frac{L}{2}\|\bTheta^{k+1}-\bTheta^k\|^2\\
    &\stackrel{\eqref{eq.SCSC-1}}{=} \bF(\bTheta^k)-\alpha_k\dotp{\nabla \bF(\bTheta^k),\nabla g(\bTheta^k;\phi^k)\nabla_1 f(\by^{k+1}, \bTheta;\xi^k)}\\
    &\ \ \  -\alpha_k\dotp{\nabla \bF(\bTheta^k), \nabla_2 f(\by^{k+1},\bTheta;\xi^k)}+\frac{L}{2}\|\bTheta^{k+1}-\bTheta^k\|^2\\
    &= \bF(\bTheta^k)-\alpha_k\|\nabla \bF(\bTheta^k)\|^2+\frac{L}{2}\|\bTheta^{k+1}-\bTheta^k\|^2\\
    &\ \ \ +\alpha_k\dotp{\nabla \bF(\bTheta^k),\nabla \dg(\bTheta^k)\nabla_1 \df(\dg(\bTheta^k), \bTheta)}\\
    &\ \ \ -\alpha_k\dotp{\nabla \bF(\bTheta^k), \nabla g(\bTheta^k;\phi^k)\nabla_1 f(\by^{k+1}, \bTheta;\xi^k)}\\
    &\ \ \  +\alpha_k\dotp{\nabla \bF(\bTheta^k), \nabla_2 \df(\dg(\bTheta^k),\bTheta)- \nabla_2 f(\by^{k+1},\bTheta,\xi^k)} .
\end{align*}
}
Conditioned on $\mathcal F^k$, taking expectation over the sampling process of $\phi^k$ and $\xi^k$ from the data set $\mathcal{M}_t \cup \mathcal{D}_t$ on both sides, we have
\begin{align*} 
    &\mathbb{E}_t\left[\bF(\bTheta^{k+1})|\mathcal F^k\right]\\
    \stackrel{(a)}{\leq}&\bF(\bTheta^k)-\alpha_k\|\nabla \bF(\bTheta^k)\|^2+\frac{L}{2}\mathbb{E}_t[\|\bTheta^{k+1}-\bTheta^k\|^2|\mathcal F^k]\\
    & +\alpha_k\left\|\nabla \bF(\bTheta^k)\right\|\,\mathbb{E}_t\left[\|\nabla g(\bTheta^k;\phi^k)\|^2|\mathcal F^k\right]^{\frac{1}{2}} \\
    &\  \times \mathbb{E}_t\left[\|\nabla_1 f(g(\bTheta^k);\bTheta^k;\xi^k)-\nabla_1 f(\by^{k+1}, \bTheta^k;\xi^k)\|^2|\mathcal F^k\right]^{\frac{1}{2}}\\
    &+\alpha_k  \left\|\nabla \bF(\bTheta^k)\right\| \\
    & \ \times \mathbb{E}_t \left[ \| \nabla_2 f(g(\bTheta^k);\bTheta^k, \xi^k)- \nabla_2 f(\by^{k+1};\bTheta^k,\xi^k)\| | \mathcal F^k \right] 
\end{align*}
\begin{align*}
    \stackrel{(b)}{\leq}& \bF(\bTheta^k)-\alpha_k\|\nabla \bF(\bTheta^k)\|^2  +LC_g^2C_{f_1}^2\alpha_k^2 + LC_{f_2}^2\alpha_k^2\\
    &+\left( \alpha_kC_gL_{f_{11}} + \alpha_k L_{f_{21}}\right)\|\nabla \bF(\bTheta^k)\| \\
    &\ \times \mathbb{E}_t\left[\|g(\bTheta^k)-\by^{k+1}\|^2|\mathcal F^k\right]^{\frac{1}{2}}\\
    \stackrel{(c)}{\leq} &\bF(\bTheta^k)-\alpha_k\|\nabla \bF(\bTheta^k)\|^2+LC_g^2C_{f_1}^2\alpha_k^2 + LC_{f_2}^2\alpha_k^2\\
    &+\frac{\alpha_k^2C_g^2L_{f_{11}}^2+\alpha_k^2L_{f_{21}}^2}{2\beta_k}\|\nabla \bF(\bTheta^k)\|^2\\
    &+\beta_k\mathbb{E}_t\left[\|g(\bTheta^k)-\by^{k+1}\|^2|\mathcal F^k\right]\\
    \leq &\,\bF(\bTheta^k)-\alpha_k\left(1-\frac{\alpha_kC_g^2L_{f_{11}}^2+\alpha_kL_{f_{21}}^2}{2\beta_k}\right)\|\nabla \bF(\bTheta^k)\|^2\\
    &+\beta_k\mathbb{E}_t\left[\|g(\bTheta^k)-\by^{k+1}\|^2|\mathcal F^k\right]+LC_g^2C_{f_1}^2\alpha_k^2 + LC_{f_2}^2\alpha_k^2,
\end{align*}
where 
in (a) we use the Cauchy-Schwartz inequality; 
 in (b) we use the update rule \eqref{update:theta}, the boundedness of $\|\nabla g\|, \|\nabla_1 f\|$ and $\|\nabla_2 f\|$ and the Lipschitz continuous gradient of $f$ from Lemma \ref{lemma-lip-stoc} in the supplemental material \ref{sec-lemma}; 
and in (c) we use the Young's inequality.

Define the Lyapunov function
\begin{align*}
  {\cal V}^k=\bF(\bTheta^k)+\|g(\bTheta^{k-1})-\by^k\|^2.
\end{align*}
It follows that
\begin{align}\label{eq.thm1-1}
    &\ \ \mathbb{E}_t[{\cal V}^{k+1}|\mathcal F^k]\\
    &\leq{\cal V}^k-\alpha_k\left(1-\frac{\alpha_kC_g^2L_{f_{11}}^2+\alpha_kL_{f_{21}}^2}{2\beta_k}\right)\|\nabla \bF(\bTheta^k)\|^2 \nonumber \\
    &\ \ \ \ +LC_g^2C_{f_1}^2\alpha_k^2 + LC_{f_2}^2\alpha_k^2 \nonumber\\
    &\ \ \ \ +(1+\beta_k)\mathbb{E}_t\left[\|g(\bTheta^k)-\by^{k+1}\|^2|\mathcal F^k\right]-\|g(\bTheta^{k-1})-\by^k\|^2\nonumber\\
    &\stackrel{(a)}{\leq} {\cal V}^k-\alpha_k\left(1-\frac{\alpha_kC_g^2L_{f_{11}}^2+\alpha_kL_{f_{21}}^2}{2\beta_k}\right)\|\nabla \bF(\bTheta^k)\|^2\nonumber \\
    &\ \ \ \ +LC_g^2C_{f_1}^2\alpha_k^2 + LC_{f_2}^2\alpha_k^2+2(1+\beta_k)\beta_k^2V_g^2\nonumber\\
    &\ \ \ \  +\left((1+\beta_k)(1-\beta_k)^2-1\right)\|g(\bTheta^{k-1})-\by^k\|^2\nonumber \\
    &\ \ \ \ +8(1+\beta_k)(1-\beta_k)^2C_g^2\left(  C_g^2 C_{f_1}^2 +  C_{f_2}^2\right)\alpha_k^2  \nonumber\\
    &\stackrel{(b)}{\leq}{\cal V}^k-\alpha_k\left(1-\frac{\alpha_kC_g^2L_{f_{11}}^2+\alpha_kL_{f_{21}}^2}{2\beta_k}\right)\|\nabla \bF(\bTheta^k)\|^2\nonumber \\
    &\ \ \ \ +LC_g^2C_{f_1}^2\alpha_k^2  + LC_{f_2}^2\alpha_k^2+2(1+\beta_k)\beta_k^2V_g^2\nonumber \\
    &\ \ \ \ +8 C_g^2\left(  C_g^2 C_{f_1}^2 +  C_{f_2}^2\right)\alpha_k^2,  \nonumber
\end{align}
where (a) follows from Lemma \ref{lemma2}, and (b) uses that $(1+\beta_k)(1-\beta_k)^2=(1-\beta_k^2)(1-\beta_k)\leq 1$. 
 The corresponding constant $V_g$ is defined in Lemma \ref{ass_bv_remark} and $L, C_g, C_{f_1}, C_{f_2}, L_{f_{11}}, L_{f_{21}}$ are defined in Lemma \ref{lemma-lip-stoc} in the supplemental material \ref{sec-lemma}.

Select  (with $\beta_k\in(0,1)$)
\begin{align} \label{eq:L0}
    \alpha_k  =\frac{\beta_k}{C_g^2L_{f_{11}}^2+L_{f_{21}}^2}:= \frac{\beta_k}{L_0}
\end{align}
so that $1-\frac{\alpha_kC_g^2L_{f_{11}}^2+\alpha_kL_{f_{21}}^2}{2\beta_k}=\frac{1}{2}$, and define 
\begin{align}\label{ctilde}
    \tilde{C}:&=  LC_g^2C_{f_1}^2   + LC_{f_2}^2 +4  \left(C_g^2L_{f_{11}}^2+L_{f_{21}}^2 \right)^2V_g^2 \nonumber\\
    &\ \ +8 C_g^2\left(  C_g^2 C_{f_1}^2 +  C_{f_2}^2\right).
\end{align}
Further taking expectation over $\mathcal F^k$ on both sides of \eqref{eq.thm1-1}, then it follows that
\begin{align}\label{eq.thm1-2}
    \mathbb{E}_t[{\cal V}^{k+1}]\leq \mathbb{E}_t[{\cal V}^k]-\frac{\alpha_k}{2}\mathbb{E}_t[\|\nabla \bF(\bTheta^k)\|^2]+\tilde{C}\alpha_k^2.
\end{align}
Telescoping over $k$ and rearranging terms, we have     
 \begin{align*}   
    \frac{\sum_{k=0}^K\alpha_k\mathbb{E}_t[\|\nabla \bF(\bTheta^k)\|^2]}{\sum_{k=0}^K\alpha_k}\leq\frac{2{\cal V}^0+2 \tilde{C}\sum_{k=0}^K\alpha_k^2}{\sum_{k=0}^K\alpha_k}.
\end{align*}
Choosing the stepsize as $\alpha_k=\frac{1}{\sqrt{K}}$ leads to 
 \begin{align*}   
    \frac{\sum_{k=0}^{K-1}\mathbb{E}_t[\|\nabla \bF(\bTheta^k)\|^2]}{K}\leq\frac{2{\cal V}^0+2 \tilde{C}}{\sqrt{K}}.
\end{align*}
By initializing of $\by^0 = g(\bTheta^{-1})$, we have 
\begin{align*}
  {\cal V}^0=\bF(\bTheta^0)+\|g(\bTheta^{-1})-\by^0\|^2 = \bF(\bTheta^0).
\end{align*}
The proof is complete. 
\end{proof}

\newpage
\onecolumn

\section{Verify Assumptions} \label{verify_ass}
In this section, we show that for power allocation problem \eqref{eq:sum_rate_IA}, Assumption \ref{ass_1} can be satisfied.
\begin{claim}
  Consider the power allocation problem \eqref{eq:sum_rate_IA} with $K$ transmitter and receiver pairs, pick $\ell(\cdot)$ and $u(\cdot)$ as suggested in  \eqref{MSE-loss} and \eqref{sumrate-loss}.  Let $\mathcal{H}$ indicate a set of bounded channel coefficients with dimension $K\times K$. Suppose that the neural network $\pi(\bTheta;\bh) \in [0, \bp_{\max}]$ has bounded Jacobian $J^\pi \in \mathbb{R}^{K\times d}$ and Hessian $H^\pi \in \mathbb{R}^{K \times d\times d}$ for all $\bTheta \in \mathbb{R}^{d}$, $\bh\in \mathcal{H}$ and $\bp \in [0, \bp_{\max} ]$ 
with $\bp_{\max} \in \mathbb{R}_{+}^K$,  where $J^{\pi}_{i,j} = \frac{\partial \pi_i(\bTheta; \bh)}{\partial \bTheta_j}$, $H^{\pi}_{k, i,j} = \frac{\partial^2 \pi_k(\bTheta; \bh)}{\partial \bTheta_i \partial \bTheta_j}$. Then Assumption \ref{ass_1} holds true, that is, there exist some positive constants {$C_{\ell_0}, C_{\ell_1}, C_{\ell_2}, C_{u_0}, C_{u_1}, C_{u_2}$}, such that the following holds:
\begin{align*}
    \left\| \ell(\bTheta; \bx, \bp) \right\| &\le C_{\ell_0}, \quad  \left\|
      u(\bTheta; \bx, \bp) \right\|  \le C_{u_0},  \; \forall~\bh\in \mathcal{H}, \forall~\bTheta\\
     \left\|\nabla_{\bTheta} \ell(\bTheta; \bx, \bp) \right\| & \le C_{\ell_1}, \quad 
    \left\|\nabla_{\bTheta} u(\bTheta; \bx, \bp) \right\| \le C_{u_1},   \; \forall~\bh\in \mathcal{H}, \forall~\bTheta\\
    \left\|\nabla_{\bTheta}^2 \ell(\bTheta; \bx, \bp)\right\| &\le C_{\ell_2}, \quad 
    \left\|\nabla_{\bTheta}^2 u(\bTheta; \bx, \bp)\right\| \le C_{u_2},\; \forall~\bh\in \mathcal{H}, \forall~\bTheta.
\end{align*}
\end{claim}

\begin{proof}
First, based on our specific problem \eqref{eq:sum_rate_IA}, we know that  the allocated power $\bp \in [0, \textbf{p}_{\max}]$, and output of the neural network $\pi(\bTheta;\bx) \in [0, \textbf{p}_{\max}]$ are bounded. Then we have: 
\begin{align*}
    \ell(\bTheta; \bx, \bp)  & =  \|\bp - \pi(\bTheta, \bx)\|^2 \in \mathbb{R},\\
    \nabla_{\bTheta} \ell(\bTheta;\bx, \bp)& = 2\left(J^{\pi(\bTheta;\bx)}\right)^T(\pi(\bTheta, \bx)-\bp ) \in \mathbb{R}^{d\times 1}, \\
    \nabla^2_{\bTheta} \ell(\bTheta;\bx, \bp) &=2 \sum_k \left[ H_{k,:,:}^{\pi(\bTheta;\bx)} (\pi_k(\bTheta, \bx) - \bp_k) \right]  + 2 (J^{\pi(\bTheta;\bx)})^T J^{\pi(\bTheta;\bx)}\in \mathbb{R}^{d\times d}.
\end{align*}
Since by assumption, we know Jacobian $J^\pi$ and Hessian $H^\pi$ are bounded, then, $\ell(\bTheta; \bx, \bp), \nabla_{\bTheta} \ell(\bTheta; \bx, \bp)$ and $\nabla_{\bTheta}^2 \ell(\bTheta; \bx, \bp)$ are all bounded. 

Similarly, we can also show that $u(\bTheta;\bx, \bp)$,  $\nabla_{\bTheta} u(\bTheta;\bx, \bp)$,  and $\nabla^2_{\bTheta} u(\bTheta;\bx, \bp)$ are bounded. Towards this end, we first compute the elements in $\nabla_{\pi} R(\pi(\bTheta, \bx); \bx)$ and $\nabla^2_{\pi} R(\pi(\bTheta, \bx); \bx)$: 
\begin{align*} 
 R(\pi(\bTheta, \bx); \bx) & =\sum_{k=1}^K  \log\left(1+\frac{|\bx_{kk}|^2 \pi_k}{\sum_{i\neq k}|\bx_{ki}|^2 \pi_i+\sigma_k^2}\right),\\
 \nabla_{\pi} R(\pi(\bTheta, \bx); \bx) & = \left[ \nabla_{\pi_k} R(\pi(\bTheta, \bx); \bx)\right]_{k=1:K},  \\
 \nabla_{\pi_k} R(\pi(\bTheta, \bx); \bx) & =  \frac{|\bx_{kk}|^2}{\sum_{i=1}^{K}|\bx_{ki}|^2 \pi_i+\sigma_k^2} + \sum_{j\ne k} \frac{- |\bx_{jj}|^2 \pi_j |\bx_{jk}|^2}{\left(\sum_{i\neq j}|\bx_{ji}|^2 \pi_i+\sigma_j^2\right) \left( \sum_{i=1}^K|\bx_{ji}|^2 \pi_i+\sigma_j^2\right)},  \\
 \nabla^2_{\pi} R(\pi(\bTheta, \bx); \bx)  & = \left[ \nabla^2_{\pi_{kt}} R(\pi(\bTheta, \bx); \bx)\right]_{k=1:K, t=1:K},  \\
   \nabla^2_{\pi_{kt}} R(\pi(\bTheta, \bx); \bx) |_{t\ne k} &= \frac{-|\bx_{kk}|^2|\bx_{kt}|^2}{\left( \sum_{i=1}^{K}|\bx_{ki}|^2 \pi_i+\sigma_k^2\right)^2} + \sum_{j \ne k, t} \frac{ |\bx_{jj}|^2 \pi_j |\bx_{jk}|^2 |\bx_{jt}|^2\left(\sum_{i\neq j}|\bx_{ji}|^2 \pi_i+\sum_{i=1}^K|\bx_{ji}|^2 \pi_i+2\sigma_j^2\right)}{\left(\sum_{i\neq j}|\bx_{ji}|^2 \pi_i+\sigma_j^2\right)^2 \left( \sum_{i=1}^K|\bx_{ji}|^2 \pi_i+\sigma_j^2\right)^2}\\
   & -  \frac{ |\bx_{tt}|^2  |\bx_{tk}|^2 \left(\sum_{i\ne t}|\bx_{ti}|^2 \pi_i+\sigma_t^2 \right)}{\left(\sum_{i\neq t}|\bx_{ti}|^2 \pi_i+\sigma_t^2\right) \left( \sum_{i=1}^K|\bx_{ti}|^2 \pi_i+\sigma_t^2\right)^2}, \\
   \nabla^2_{\pi_{kt}} R(\pi(\bTheta, \bx); \bx) |_{t= k} &= \frac{-|\bx_{kk}|^2|\bx_{kt}|^2}{\left( \sum_{i=1}^{K}|\bx_{ki}|^2 \pi_i+\sigma_k^2\right)^2} + \sum_{j \ne k} \frac{ |\bx_{jj}|^2 \pi_j |\bx_{jk}|^2 |\bx_{jt}|^2\left(\sum_{i\neq j}|\bx_{ji}|^2 \pi_i+\sum_{i=1}^K|\bx_{ji}|^2 \pi_i+2\sigma_j^2\right)}{\left(\sum_{i\neq j}|\bx_{ji}|^2 \pi_i+\sigma_j^2\right)^2 \left( \sum_{i=1}^K|\bx_{ji}|^2 \pi_i+\sigma_j^2\right)^2}.
\end{align*}
Because all elements of $\bh_{ij}$, $\pi_i$, and $\sigma_k$ are bounded,  it is clear that both $\nabla_{\pi} R(\pi(\bTheta, \bx); \bx)$ and $\nabla^2_{\pi} R(\pi(\bTheta, \bx); \bx)$ are bounded. 

Next, we derive the bounds for $u(\bTheta;\bx, \bp)$,  $\nabla_{\bTheta} u(\bTheta;\bx, \bp)$,  and $\nabla^2_{\bTheta} u(\bTheta;\bx, \bp)$,
\begin{align*} 
 u(\bTheta;\bx, \bp) & = R(\pi(\bTheta, \bx); \bx)=\sum_{k=1}^K \alpha_k\log\left(1+\frac{|\bx_{kk}|^2 \pi_k}{\sum_{j\neq k}|\bx_{kj}|^2 \pi_j+\sigma_k^2}\right),\\
\nabla_{\bTheta} u(\bTheta;\bx, \bp) & =  (J^{\pi(\bTheta;\bx)})^T \cdot \nabla_{\pi} R(\pi(\bTheta, \bx); \bx),\\
 \nabla^2_{\bTheta}  u(\bTheta;\bx, \bp) & = \sum_k \left[ (H_{k,:,:}^{\pi(\bTheta;\bx)})^T \cdot \nabla_{\pi_k} R(\pi(\bTheta, \bx); \bx)\right] + (J^{\pi(\bTheta;\bx)})^T \cdot \nabla^2_{\pi} R(\pi(\bTheta, \bx); \bx) \cdot J^{\pi(\bTheta;\bx)},
\end{align*}
given all of $J^\pi$, $H^\pi$, $\nabla_{\pi} R(\pi(\bTheta, \bx); \bx)$ and $\nabla^2_{\pi} R(\pi(\bTheta, \bx); \bx)$ are bounded, the proof is complete.
Finally, we note that the assumption of boundedness of Jacobian $ J^\pi $ and Hessian $ H^\pi $ are reasonable; {see for example \cite[Theorem 3.2]{herrera2020estimating}, where the Lipschitz continuous and Lipschitz continuous gradient constants of neural networks are explicitly characterized.}
\end{proof}

\section{Proof of Lemma \ref{lemma2}}

\begin{claim}[Tracking Error Contraction  {\cite[Lemma 1]{chen2020solving}}]
Consider ${\cal F}^k$ as the collection of random variables, i.e., ${\cal F}^k:=\left\{\phi^0, \ldots, \phi^{k-1}, \xi^0, \ldots, \xi^{k-1}\right\}$. 
Suppose Assumption \ref{ass_unbiased} and \ref{ass_1}  hold, and $\by^{k+1}$ is generated by running  iteration \eqref{eq.SCSC-2}  conditioned ${\cal F}^k$. The mean square error of $\by^{k+1}$ satisfies
\begin{align*} 
&\mathbb{E}_t\left[\|\dg(\bTheta^k)-\by^{k+1}\|^2\mid{\cal F}^k\right] \leq  (1-\beta_k)^2\|\dg(\bTheta^{k-1})-\by^k\|^2 +4(1-\beta_k)^2C_g^2\|\bTheta^k-\bTheta^{k-1}\|^2+2\beta_k^2V_g^2,
\end{align*}
where $C_g$ and $V_g$ are defined in Lemma \ref{ass_bv_remark}.
\end{claim}

\begin{proof} \label{eq.lemma2-prof}
From the update \eqref{eq.SCSC-2}, we have that
\begin{align}\label{eq.pflemma2-1}
\by^{k+1}-\dg(\bTheta^k)&=(1-\beta_k)(\by^k-\dg(\bTheta^{k-1}))+(1-\beta_k)(\dg(\bTheta^{k-1})-\dg(\bTheta^k))+\beta_k(g(\bTheta^k;\phi^k)-\dg(\bTheta^k))\nonumber\\
&\quad+(1-\beta_k)(g(\bTheta^k;\phi^k)-g(\bTheta^{k-1};\phi^k))\nonumber\\
&=(1-\beta_k)(\by^k-\dg(\bTheta^{k-1}))+(1-\beta_k)T_1+\beta_k T_2+(1-\beta_k)T_3,
\end{align}
where we define the three terms as
\begin{align*}
    T_1:=\dg(\bTheta^{k-1})-\dg(\bTheta^k) \quad T_2:=g(\bTheta^k;\phi^k)-\dg(\bTheta^k) \quad 
    T_3:=g(\bTheta^k;\phi^k)-g(\bTheta^{k-1};\phi^k).
\end{align*}
Conditioned on ${\cal F}^k$, taking expectation over the sampling process of $\phi^k$   from the data set $\mathcal{M}_t \cup \mathcal{D}_t$, we have
\begin{align} \label{zero_exp}
\mathbb{E}_t\left[(1-\beta_k)T_1+\beta_k T_2+(1-\beta_k)T_3|\mathcal F^k\right]=\mathbf{0}\quad  \text{and}\quad  \mathbb{E}_t\left[ T_2 |\mathcal F^k\right]=\mathbf{0}.
\end{align}
Therefore,  taking a conditional expectation on the norm square of both sides of \eqref{eq.pflemma2-1}, we have 
\begin{align*}
\,&\mathbb{E}_t[\|\by^{k+1}-\dg(\bTheta^k)\|^2|\mathcal F^k]\\
\stackrel{\eqref{eq.pflemma2-1}}{=} &\mathbb{E}_t[\|(1-\beta_k)(\by^k-\dg(\bTheta^{k-1}))\|^2|\mathcal F^k]+\mathbb{E}_t\left[\|(1-\beta_k)T_1+\beta_k T_2+(1-\beta_k)T_3\|^2|\mathcal F^k\right]\\
&+2\mathbb{E}_t\left[\left\langle (1-\beta_k)(\by^k-\dg(\bTheta^{k-1})), (1-\beta_k)T_1+\beta_k T_2+(1-\beta_k)T_3\right\rangle|\mathcal F^k\right]\\
\stackrel{\eqref{zero_exp}}{=}  &(1-\beta_k)^2\|\by^k-\dg(\bTheta^{k-1})\|^2+\mathbb{E}_t\left[\|(1-\beta_k)T_1+\beta_k T_2+(1-\beta_k)T_3\|^2|\mathcal F^k\right]\\
\stackrel{\rm{(i)}}{\leq}  & (1-\beta_k)^2\|\by^k-\dg(\bTheta^{k-1})\|^2+2\mathbb{E}_t\left[ \|(1-\beta_k)T_1+\beta_kT_2\|^2|\mathcal F^k\right]+2(1-\beta_k)^2\mathbb{E}_t\left[\|T_3\|^2|\mathcal F^k\right]\\
\stackrel{}{=} &(1-\beta_k)^2\|\by^k-\dg(\bTheta^{k-1})\|^2+2(1-\beta_k)^2\mathbb{E}_t[\|T_1\|^2\mid{\cal F}^k]+2\beta_k^2\mathbb{E}_t[\|T_2\|^2\mid{\cal F}^k]\\
&+4\beta_k(1-\beta_k)\left\langle T_1,\mathbb{E}_t[T_2\mid{\cal F}^k]\right\rangle+2(1-\beta_k)^2\mathbb{E}_t\left[\|T_3\|^2|\mathcal F^k\right]\\
\stackrel{\rm{(ii)}}{\leq} & (1-\beta_k)^2\|\by^k-\dg(\bTheta^{k-1})\|^2+2(1-\beta_k)^2\mathbb{E}_t\left[\|\dg(\bTheta^k)-\dg(\bTheta^{k-1})\|^2|\mathcal F^k\right]\\
&+2(1-\beta_k)^2\mathbb{E}_t\left[\|g(\bTheta^k;\phi^k)-g(\bTheta^{k-1};\phi^k)\|^2|\mathcal F^k\right]+2\beta_k^2V_g^2\\
\stackrel{\rm{(iii)}}{\leq} & (1-\beta_k)^2\|\by^k-\dg(\bTheta^{k-1})\|^2+4(1-\beta_k)^2C_g^2\|\bTheta^k-\bTheta^{k-1}\|^2+2\beta_k^2V_g^2,
\end{align*}
where in $\rm{(i)}$ we use the Cauchy–Schwartz inequality, in $\rm{(ii)}$ we use the  bounded variance  property from Lemma \ref{ass_bv_remark} and the unbiasedness \eqref{zero_exp}, and in $\rm{(iii)}$ we use the property that the $\dg(\bTheta)$ and $g(\bTheta; \phi)$ are Lipschitz continuous from \eqref{eq.ass2-g}.
The proof is then complete.
\end{proof}

\section{Additional Lemmas } \label{sec-lemma}
\begin{lemma} \label{lemma-lip}
Suppose Assumption \ref{ass_1} holds, then function $\bF(\cdot)$ has Lipshictz continuous gradient with some universal constant $\bar{L}$, where
\begin{align*}
\dL:=  \dC_{f_1} \dL_g+ \dC_g^2 \dL_{f_{11}} + \dC_g \dL_{f_{12}} +\dC_g \dL_{f_{21}} + \dL_{f_{22}},
\end{align*}
and
\begin{align*}
& \dL_{f_{11}} :=2C_{l_0}e^{4C_{u_0}},  \quad  \dL_{f_{12}}=\dL_{f_{21}}:=C_{u_1}C_{l_0}e^{3C_{u_0}} + C_{l_1}e^{3C_{u_0}}, \\
&\dL_{f_{22}}:= \left(C_{u_1}^2 C_{\ell_0} + C_{u_2} C_{\ell_0} + 2C_{u_1} C_{\ell_1} + C_{\ell_2}\right) e^{2C_{u_0}}, \quad \dL_g:= C_{u_1}^2e^{C_{u_0}} +   C_{u_2} e^{C_{u_0}},\\
& \bar{C}_{f_1}:= C_{l_0}e^{3C_{u_0}}, \quad \bar{C}_{f_2}:=  C_{u_1}C_{l_0} e^{2C_{u_0}} + C_{l_1}e^{2C_{u_0}}, \quad \bar{C}_g:= C_{u_1}e^{C_{u_0}}.
\end{align*}
\end{lemma}

\begin{proof}
To begin with, we show that $\df$ and $\dg$ are bounded. Given  Assumption \ref{ass_1} and $\df, \dg$ are defined as \eqref{eq-relation}, we have 
\begin{subequations}   
\begin{align*}
\dg(\bTheta)&:= \frac{1}{|\mathcal{D}|} \cdot \sum_{i \in \mathcal{D}} e^{u(\bTheta; \bx^{(i)}, \bp^{(i)})} \in \left[ e^{-C_{u_0}}, e^{C_{u_0}}\right],\\
\df(\dg(\bTheta');\bTheta)&:= \frac{\sum_{i \in \mathcal{D}} e^{u(\bTheta; \bx^{(i)}, \bp^{(i)})} \cdot \ell(\bTheta; \bx^{(i)}, \bp^{(i)})}{|\mathcal{D}| \cdot \dg(\bTheta')} \in  \left[ -C_{l_0}e^{2C_{u_0}},  C_{l_0}e^{2C_{u_0}} \right],
\end{align*}
\end{subequations}
where $\mathcal{D}:= \mathcal{M}_t \cup \mathcal{D}_t$, and $C_{\ell_0}, C_{u_0}$ are defined in xxx. Note that we abused the notation a bit by omitting the subscript pf $t$ when defining $\mathcal{D}$. 

Next, we show that the gradients of $\df$ and $\dg$ are bounded. We obtain their gradients as following 
\begin{align*} 
\nabla  \dg(\bTheta)  &= \frac{1}{|\mathcal{D}|} \cdot \sum_{i \in \mathcal{D}}\left( e^{u(\bTheta; \bx^{(i)}, \bp^{(i)})} \cdot  \nabla u(\bTheta; \bx^{(i)}, \bp^{(i)}) \right),\\
\nabla_1 \df(z;\bTheta) &=- \frac{\sum_{i \in \mathcal{D}} e^{u(\bTheta; \bx^{(i)}, \bp^{(i)})} \cdot \ell(\bTheta; \bx^{(i)}, \bp^{(i)})}{|\mathcal{D}|\cdot z^2}, \\
\nabla_2 \df(z;\bTheta) &=  \frac{\sum_{i \in \mathcal{D}} 
\left( e^{u(\bTheta; \bx^{(i)}, \bp^{(i)})} \cdot  \nabla u(\bTheta; \bx^{(i)}, \bp^{(i)}) \cdot \ell(\bTheta; \bx^{(i)}, \bp^{(i)})  \right) }{|\mathcal{D}|\cdot z } \\
& + \frac{\sum_{i \in \mathcal{D}} 
\left(  e^{u(\bTheta; \bx^{(i)}, \bp^{(i)})} \cdot \nabla \ell(\bTheta; \bx^{(i)}, \bp^{(i)}) \right) }{|\mathcal{D}|\cdot z }.
\end{align*}
Combining the boundedness property from Assumption \ref{ass_1}, we conclude that for any $\bTheta$ and $\bTheta'$,
\begin{align} \label{det-bounded} 
\begin{split}
\left\| \nabla  \dg(\bTheta)\right\|  &\le C_{u_1}e^{C_{u_0}}  :=\bar{C}_g,\\
\left\|\nabla_1 \df(\dg(\bTheta');\bTheta) \right\| & \le C_{l_0}e^{3C_{u_0}}:=\bar{C}_{f_1},\\
\left\|\nabla_2 \df(\dg(\bTheta');\bTheta)\right\|  &\le  C_{u_1}C_{l_0} e^{2C_{u_0}} + C_{l_1}e^{2C_{u_0}}:=\bar{C}_{f_2}.
\end{split}
\end{align}
Next, we show that functions $\nabla \df$ and $\nabla \dg$ are $L_f$- and $L_g$-smooth by bounding the Hessian of $\|\df\|$ and $\|\dg\|$, where we have
\begin{align*} 
\nabla_{11}^2  \df(\bz;\bTheta) &=  \frac{2 \cdot \sum_{i \in \mathcal{D}} e^{u(\bTheta; \bx^{(i)}, \bp^{(i)})} \cdot \ell(\bTheta; \bx^{(i)}, \bp^{(i)})}{|\mathcal{D}| \cdot \bz^3}, \\
\nabla_{12}^2 \df(\bz;\bTheta) &=-\frac{\sum_{i \in \mathcal{D}} e^{u(\bTheta; \bx^{(i)}, \bp^{(i)})} \cdot \nabla u (\bTheta; \bx^{(i)}, \bp^{(i)}) \cdot \ell(\bTheta; \bx^{(i)}, \bp^{(i)})}{|\mathcal{D}| \cdot\bz^2}
\\
&- \frac{\sum_{i \in \mathcal{D}} e^{u(\bTheta; \bx^{(i)}, \bp^{(i)})} \cdot \nabla\ell(\bTheta; \bx^{(i)}, \bp^{(i)})}{|\mathcal{D}| \cdot\bz^2}, \\
\nabla_{21}^2 \df(\bz;\bTheta) &=  - \frac{\sum_{i \in \mathcal{D}} 
\left( e^{u(\bTheta; \bx^{(i)}, \bp^{(i)})} \cdot  \nabla u(\bTheta; \bx^{(i)}, \bp^{(i)}) \cdot \ell(\bTheta; \bx^{(i)}, \bp^{(i)})  \right) }{|\mathcal{D}| \cdot \bz^2 } \\
& - \frac{\sum_{i \in \mathcal{D}} 
\left(  e^{u(\bTheta; \bx^{(i)}, \bp^{(i)})} \cdot \nabla \ell(\bTheta; \bx^{(i)}, \bp^{(i)}) \right) }{|\mathcal{D}| \cdot \bz^2}, 
\end{align*}
\begin{align*}
\nabla_{22}^2 \df(\bz;\bTheta) &=  \frac{\sum_{i \in \mathcal{D}}
\left( e^{u(\bTheta; \bx^{(i)}, \bp^{(i)})} \cdot  \left\|\nabla u(\bTheta; \bx^{(i)}, \bp^{(i)}) \right\|^2 \cdot \ell(\bTheta; \bx^{(i)}, \bp^{(i)})  \right) }{|\mathcal{D}| \cdot\bz } \\
&+\frac{\sum_{i \in \mathcal{D}} 
\left( e^{u(\bTheta; \bx^{(i)}, \bp^{(i)})} \cdot  \nabla^2 u(\bTheta; \bx^{(i)}, \bp^{(i)}) \cdot \ell(\bTheta; \bx^{(i)}, \bp^{(i)})  \right) }{|\mathcal{D}| \cdot\bz } \\
&+\frac{2 \cdot \sum_{i \in \mathcal{D}} 
\left( e^{u(\bTheta; \bx^{(i)}, \bp^{(i)})} \cdot  \nabla u(\bTheta; \bx^{(i)}, \bp^{(i)}) \cdot \nabla \ell(\bTheta; \bx^{(i)}, \bp^{(i)})  \right) }{|\mathcal{D}| \cdot\bz } \\
& + \frac{\sum_{i \in \mathcal{D}} 
\left(  e^{u(\bTheta; \bx^{(i)}, \bp^{(i)})} \cdot \nabla^2 \ell(\bTheta; \bx^{(i)}, \bp^{(i)}) \right) }{|\mathcal{D}| \cdot\bz }, \\
\nabla^2  \dg(\bTheta)  &=\frac{1}{|\mathcal{D}|} \cdot \sum_{i \in \mathcal{D}}\left( e^{u(\bTheta; \bx^{(i)}, \bp^{(i)})} \cdot  \left\|\nabla u(\bTheta; \bx^{(i)}, \bp^{(i)})\right\|^2 + e^{u(\bTheta; \bx^{(i)}, \bp^{(i)})} \cdot  \nabla^2 u(\bTheta; \bx^{(i)}, \bp^{(i)}) \right). 
\end{align*}
Further considering that the Assumption \ref{ass_1}, we can conclude that for any $\bTheta$ and $\bTheta'$
\begin{align*}
\left\| \nabla_{11}^2  \df(\dg(\bTheta');\bTheta) \right\| &\le 2C_{l_0}e^{4C_{u_0}}:= \dL_{f_{11}}, \\
\left\| \nabla_{12}^2 \df(\dg(\bTheta');\bTheta) \right\| &\le   C_{u_1}C_{l_0}e^{3C_{u_0}} + C_{l_1}e^{3C_{u_0}}:=\dL_{f_{12}}, \\
\left\| \nabla_{21}^2 \df(\dg(\bTheta');\bTheta) \right\| &\le  C_{u_1}C_{l_0}e^{3C_{u_0}} + C_{l_1}e^{3C_{u_0}} :=\dL_{f_{21}},\\
\left\| \nabla_{22}^2 \df(\dg(\bTheta');\bTheta)  \right\| &\le   \left(C_{u_1}^2 C_{\ell_0} + C_{u_2} C_{\ell_0} + 2C_{u_1} C_{\ell_1} + C_{\ell_2}\right) e^{2C_{u_0}} :=\dL_{f_{22}}, \\
\left\| \nabla^2  \dg(\bTheta)  \right\| &\le   C_{u_1}^2e^{C_{u_0}} +   C_{u_2} e^{C_{u_0}}:=\dL_g.
\end{align*}
Or in other words, when $z = \dg(\cdot)$, the following holds true: 
\begin{align} \label{det-lip}
\begin{split}
\|\nabla_1 \df(\bz, \cdot)-\nabla_1 \df(\bz', \cdot)\| & \leq \dL_{f_{11}}\|\bz-\bz'\|, \\
\|\nabla_1 \df(\bz, \bTheta)-\nabla_1 \df(\bz, \bTheta')\| & \leq \dL_{f_{12}}\|\bTheta-\bTheta'\|, \\
\|\nabla_2 \df(\bz, \cdot)-\nabla_2 \df(\bz', \cdot)\| & \leq \dL_{f_{21}}\|\bz-\bz'\|, \\
\|\nabla_2 \df(\bz, \bTheta)-\nabla_2 \df(\bz, \bTheta')\| & \leq \dL_{f_{22}}\|\bTheta-\bTheta'\|, \\
\|\nabla \dg(\bTheta)-\nabla \dg(\bTheta')\| & \leq \dL_g\|\bTheta-\bTheta'\|.
\end{split}
\end{align}
Then, we are ready to establish the smoothness condition of the gradient of the compositional function $\bF(\bTheta)=\df(\dg(\bTheta), \bTheta)$, we have
\begin{align*} 
    \nabla \bF(\bTheta)= &\nabla \dg(\bTheta) \nabla_1 \df(\dg(\bTheta), \bTheta)  + \nabla_2 \df(\dg(\bTheta), \bTheta),
\end{align*}
and
\begin{align*}
 \|\nabla \bF(\bTheta) - \nabla \bF(\bTheta') \| 
\le & \|\nabla \dg(\bTheta) \nabla_1 \df(\dg(\bTheta), \bTheta)  -\nabla \dg(\bTheta') \nabla_1 \df(\dg(\bTheta), \bTheta)  \| \\
& + \|\nabla \dg(\bTheta') \nabla_1 \df(\dg(\bTheta), \bTheta)  -\nabla \dg(\bTheta') \nabla_1 \df(\dg(\bTheta'), \bTheta')  \| \\
& +\|  \nabla_2 \df(\dg(\bTheta), \bTheta) - \nabla_2 \df(\dg(\bTheta'), \bTheta')\|\\
\le & \bar{C}_{f_1}   \dL_g\|\bTheta-\bTheta'\| + \bar{C}_{g}   \dL_{f_{11}}\|\dg(\bTheta)-\dg(\bTheta')\| + \bar{C}_{g} \dL_{f_{12}}\|\bTheta-\bTheta'\| \\
& +  \dL_{f_{21}}\|\dg(\bTheta)-\dg(\bTheta')\| +   \dL_{f_{22}}\|\bTheta-\bTheta'\| \\
\le & \dL \|\bTheta-\bTheta'\|,
\end{align*}
 where in the second inequality we use the boundedness of $\|\nabla \dg\|, \|\nabla \df\|$ and the Lipschitz continuous gradient of $\dg$ and $\df$, see \eqref{det-bounded} and \eqref{det-lip}, and in the third inequality we use the Lipschitz continuity of $\dg$ and $\df$ (implied by the boundedness of $\|\nabla \dg\|, \|\nabla \df\|$), and $\dL$ is defined as
\begin{align*}
\dL:=  \dC_{f_1} \dL_g+ \dC_g^2 \dL_{f_{11}} + \dC_g \dL_{f_{12}} +\dC_g \dL_{f_{21}} + \dL_{f_{22}}.
\end{align*}
The proof is complete.
\end{proof}

Then, follows the same reasoning, we will have following results when stochastic sampling is used,
\begin{lemma} \label{lemma-lip-stoc}
Suppose Assumption  \ref{ass_1} and \ref{ass_unbiased}  hold,  $f(\cdot)$ and $g(\cdot)$ are defined as \eqref{eq-relation-sample}, and the first input of $f(\bz;\bTheta; \xi)$ is  bounded away from zero as $\|\bz\|\ge C_z$, then the following holds:

(1) The stochastic gradients of $f$ and $g$ are bounded in expectation, that is, there exist positive constants $C_g, C_{f_1}, C_{f_2}$, such that the following relations hold
\begin{align*} 
\begin{split}
  \mathbb{E}_t\left[\|\nabla g(\bTheta;\phi)\|\right] &\leq C_g, \\
 \mathbb{E}_t\left[\|\nabla_1 f(\bz;\bTheta;\xi)\|\right] & \leq C_{f_1}, \\
\mathbb{E}_t\left[\|\nabla_2 f(\bz;\bTheta;\xi)\|\right] &\leq C_{f_2},  
\end{split}
\end{align*}
where  the constants are defined as: 
\begin{align*}
& {C}_{f_1}:=  C_{l_0}e^{C_{u_0}} / C_z^2, \quad {C}_{f_2}:=   C_{u_1}C_{l_0} e^{C_{u_0}} / C_z+  C_{l_1}e^{C_{u_0}}/C_z, \quad {C}_g:= C_{u_1}e^{C_{u_0}}.
\end{align*}

(2) Functions $\nabla f$ and $\nabla g$ are $L_f$- and $L_g$-smooth, that is, for any $\bTheta, \bTheta'\in\mathbb{R}^d$, and $\bz, \bz'$ satisfying $\|\bz\|\le C_{z}$ and $\|\bz'\|\le C_z$, we have: 
\begin{align}\label{eq.ass1}
\begin{split}
\|\nabla_1 f(\bz, \bTheta;\xi)-\nabla_1 f(\bz', \bTheta;\xi)\| & \leq L_{f_{11}}\|\bz-\bz'\|, \\
\|\nabla_1 f(\bz, \bTheta;\xi)-\nabla_1 f(\bz, \bTheta';\xi)\| & \leq L_{f_{12}}\|\bTheta-\bTheta'\|, \\
\|\nabla_2 f(\bz, \bTheta;\xi)-\nabla_2 f(\bz', \bTheta;\xi)\| & \leq L_{f_{21}}\|\bz-\bz'\|, \\
\|\nabla_2 f(\bz, \bTheta;\xi)-\nabla_2 f(\bz, \bTheta';\xi)\| & \leq L_{f_{22}}\|\bTheta-\bTheta'\|, \\
\|\nabla g(\bTheta;\phi)-\nabla g(\bTheta';\phi)\| & \leq L_g\|\bTheta-\bTheta'\|,
\end{split}
\end{align}
where  
\begin{align*}
& L_{f_{11}} :=2 C_{l_0}e^{C_{u_0}}/C_z^3,  \quad  L_{f_{12}}=L_{f_{21}}:= {\left(C_{u_1}C_{l_0}e^{C_{u_0}} +  C_{l_1}e^{C_{u_0}}\right)}/{C_z^2} , \\
&L_{f_{22}}:= \left(C_{u_1}^2 C_{\ell_0} + C_{u_2} C_{\ell_0} + 2C_{u_1} C_{\ell_1} + C_{\ell_2}\right)  e^{C_{u_0}} / C_z, \quad L_g:= C_{u_1}^2e^{C_{u_0}} +   C_{u_2} e^{C_{u_0}}.
\end{align*}
\end{lemma}

\begin{proof}
The derivation of this result is similar to those presented in Lemma \ref{lemma-lip}, except the change of $i \in \mathcal{D}$ with $i \in \xi$ or $i \in \phi$ in derivations, and the usage of  $z=\by$ instead of $z=\dg (\cdot)$.
\end{proof}

\end{document}